\documentclass[12pt]{article}

\usepackage{amssymb,amsmath,amsfonts,eurosym,geometry,ulem,graphicx,caption,color,setspace,sectsty,comment,footmisc,caption,natbib,pdflscape,subfigure,array,hyperref,mathtools,multirow,enumitem,microtype}
\mathtoolsset{showonlyrefs,showmanualtags}
\usepackage[bibliography=common]{apxproof}
\newtheorem{thm}{Theorem}
\newtheorem*{thmkp}{Theorem 1.1 in Kim and Pollard (1990)}
\newtheorem{lem}[thm]{Lemma}
\newtheorem{corollary}{Corollary}[thm]
\newtheorem{ass}{Assumption}
\newtheorem{rem}{Remark}

\newcommand{\R}{\mathbb{R}}

\newcolumntype{L}[1]{>{\raggedright\let\newline\\arraybackslash\hspace{0pt}}m{#1}}
\newcolumntype{C}[1]{>{\centering\let\newline\\arraybackslash\hspace{0pt}}m{#1}}
\newcolumntype{R}[1]{>{\raggedleft\let\newline\\arraybackslash\hspace{0pt}}m{#1}}

\newcommand{\E}{\mathbb{E}}
\newcommand{\ind}{\mathbf{1}}
\newcommand{\argmax}{\mathop{\rm arg~max}\limits}
\newcommand{\argmin}{\mathop{\rm arg~min}\limits}
\definecolor{green(pigment)}{rgb}{0.0, 0.65, 0.31}

\definecolor{cornellred}{rgb}{0.7, 0.11, 0.11}

\DeclareMathOperator{\plim}{\mathrm{plim}}

\DeclareMathOperator{\var}{\mathrm{Var}}

\begin{document}

\begin{titlepage}
\title{Regret Analysis in Threshold Policy Design\thanks{I am grateful to Ivan Canay, Joel Horowitz, and Charles Manski for their guidance in this project. I am also thankful to Anastasiia Evdokimova, Amilcar Velez, the Editor, the Associate Editor, the referees, participants of the Econometric Reading Group at Northwestern University, and attendees of the 2024 Annual Health Econometrics Workshop at John Hopkins University for their comments and suggestions.}}
\author{Federico Crippa \\ Department of Economics \\ Northwestern University\\
\href{mailto:federicocrippa2025@u.northwestern.edu}{federicocrippa2025@u.northwestern.edu}
}
\date{April 4, 2025}
\maketitle
\begin{abstract}

Threshold policies are decision rules that assign treatments based on whether an observable characteristic exceeds a certain threshold. They are widespread across multiple domains, including welfare programs, taxation, and clinical medicine. This paper examines the problem of designing threshold policies using experimental data, when the goal is to maximize the population welfare. First, I characterize the regret -- a measure of policy optimality -- of the Empirical Welfare Maximizer (EWM) policy, popular in the literature. Next, I introduce the Smoothed Welfare Maximizer (SWM) policy, which improves the EWM's regret convergence rate under an additional smoothness condition. The two policies are compared by studying how differently their regrets depend on the population distribution, and investigating their finite sample performances through Monte Carlo simulations. In many contexts, the SWM policy guarantees larger welfare than the EWM. An empirical illustration demonstrates how the treatment recommendations of the two policies may differ in practice.

\noindent \\
\vspace{0in}

\bigskip
\end{abstract}
\setcounter{page}{0}
\thispagestyle{empty}
\end{titlepage}
\pagebreak \newpage

\doublespacing

\section{Introduction}

Treatments are rarely universally assigned. When their effects are heterogeneous across individuals, policymakers aim to target those who would benefit the most from specific interventions. Scholarships, for example, are awarded to students with high academic performance or financial need; tax credits are provided to companies engaged in research and development activities; medical treatments are prescribed to sick patients. Despite the potential complexity and multidimensionality of heterogeneous treatment effects, treatment eligibility criteria are often kept quite simple. This paper studies one of the most common of these simple assignment mechanisms: threshold policies, where the decision to assign the treatment is based on whether a scalar observable characteristic — referred to as the \textit{index} — exceeds a specified \textit{threshold}.

Threshold policies are ubiquitous, ranging across multiple domains. In welfare policies, they regulate the qualification for public health insurance programs through age \citep{card2008impact,shigeoka2014effect} and anti-poverty programs through income \citep{crost2014aid}. In taxation, they determine marginal rates through income brackets \citep{taylor2003corporation}. In clinical medicine, the referral for liver transplantation depends on whether a composite of laboratory values obtained from blood tests is beyond a certain threshold \citep{kamath2007model}. Even criminal offenses are defined through threshold policies: sanctions for Driving Under the Influence are based on whether the Blood Alcohol Content exceeds specific values.

Economists frequently study the outcomes of threshold policies, with the regression discontinuity design (RDD) being a widely used tool for causal inference. The RDD focuses on an ex-post evaluation of treatment effects at discontinuity points. In this paper, my perspective is different: I consider the ex-ante problem faced by a policymaker seeking to implement a threshold policy and interested in maximizing the average social welfare, targeting individuals who would benefit from the treatment. Experimental data are available: how should they be used to design the threshold policy for the population?

Answering this question requires defining a criterion by which policies are evaluated. Since the performance of a policy depends on the unknown data distribution, the policymaker searches for a policy that behaves uniformly well across a specified family of data distributions (the state space). The regret of a policy is the (possibly) random difference between the maximum achievable welfare and the welfare it generates in the population. Policies can be evaluated considering their maximum expected regret \citep{manski2004statistical, hirano2009asymptotics, kitagawa2018should, athey2021policy, mbakop2021model, manski2023probabilistic}, or other worst-case statistics of the regret distribution \citep{manski2023statistical, kitagawa2022treatment}. Once the criterion has been established, optimal policy learning aims to pinpoint the policy that minimizes it. Rather than directly tackling the functional minimization problem, following the literature, I consider candidate threshold policy functions and characterize some properties of their regret.

The first contribution of this paper is to show how to derive the asymptotic distribution of the regret for a given threshold policy. The underlying intuition is simple: threshold policies use sample data to choose the threshold, which is hence a random variable with a certain asymptotic behavior. A Taylor expansion establishes a map between the regret of a policy and its threshold, allowing one to characterize the asymptotic distribution of the regret through the asymptotic behavior of the threshold. This shifts the problem to characterizing the asymptotic distribution of the threshold estimator, simplifying the analysis as threshold estimators can be studied with common econometric tools.

I start considering the Empirical Welfare Maximizer (EWM) policy studied by \cite{kitagawa2018should}. They derive uniform bounds for the expected regret of the policy for various policy classes, where the policy class impacts the findings only in terms of its VC dimensionality. My approach is more specific, considering only threshold policies, but also more informative: leveraging the knowledge of the policy class, I characterize the asymptotic distribution of the regret. As mentioned above, this requires the derivation of the asymptotic distribution of the threshold for the EWM policy, which is non-standard: it exhibits the ``cube root asymptotics'' behavior studied in \cite{kim1990cube}. The convergence rate is $n^{\frac{1}{3}}$, and the asymptotic distribution is of \cite{chernoff1964estimation} type. The non-standard behavior and the unusual convergence rate are due to the discontinuity in the objective function and are reflected in the asymptotic distribution of the regret, and in its $n^{\frac{2}{3}}$ pointwise convergence rate.

My second contribution is hence the introduction of a novel threshold policy, the Smoothed Welfare Maximizer (SWM) policy. This approach modifies the Empirical Welfare Maximizer (EWM) policy by replacing the indicator function in its objective function with a smooth kernel. Under certain smoothness regularity assumptions, the threshold estimator for the SWM policy is asymptotically normal, and its regret achieves a pointwise convergence rate of $n^{\frac{4}{5}}$. While this paper does not establish the optimality of this rate -- meaning it does not determine whether the SWM policy attains the fastest possible convergence rate under the given regularity conditions -- my results imply that, under the additional assumption, the SWM policy achieves a faster pointwise convergence rate than the EWM policy.

Building on these asymptotic results, I extend the comparison of the regrets with the EWM and the SWM policies beyond their convergence rates. My findings allow to compare the asymptotic distributions and investigate how differently they depend on the data distribution; theoretical results are helpful to inform and guide the Monte Carlo simulations, which confirm that the asymptotic results approximate the actual finite sample behaviors. Notably, the simulations confirm that the SWM policy may guarantee lower expected regret in finite samples.

To demonstrate the practical differences between the two policies, I present an empirical illustration considering a job-training treatment. In that context, the SWM threshold policy would recommend treating 66.2\% of unemployed workers, as opposed to 63.6\% with the EWM policy. This difference of almost 3 percentage points is economically non-negligible.

\subsection{Related Literature}

This paper relates to the statistical decision theory literature studying the problem of policy assignment with covariates \citep{manski2004statistical,stoye2012minimax,kitagawa2018should,athey2021policy,mbakop2021model,sun2021treatment,sun2021empirical,viviano2023fair}. My setting is mainly related to \cite{kitagawa2018should} and \cite{athey2021policy}, with some notable differences. \cite{kitagawa2018should} study the EWM policy for policy classes with finite VC dimension. They derive finite sample bounds for the expected regret without relying on smoothness assumptions. \cite{athey2021policy} consider a double robust version of the EWM and allow for observational data. Under smoothness assumptions analogous to mine, they derive asymptotic bounds for the expected regret for policy classes with finite VC dimensions. Conversely, results in this paper apply exclusively to threshold policies, relying on a combination of the assumptions in \cite{kitagawa2018should} and \cite{athey2021policy}. The narrower focus allows for more comprehensive results: I derive the asymptotic distribution of the regret, rather than providing some bounds for the expected regret. A critical distinction lies in the different nature of the convergence rates. My results are valid pointwise, derived by leveraging additional assumptions on the data distribution. As a result, the rates I obtain for the EWM and SWM policies are faster than the $\sqrt{n}$ rate reported as optimal by \cite{kitagawa2018should} and \cite{athey2021policy}. Their $\sqrt{n}$ rate is, in fact, uniformly valid for a broader family of data distributions, including extreme cases (e.g., where conditional ATE is flat at the threshold) that are excluded from my analysis. Their uniform results may be viewed as a benchmark: when more structure is imposed on the problem and certain data distributions are excluded, the rates can be improved.

Optimal policy learning finds its empirical counterpart in the literature dealing with targeting, especially common in development economics. Recent studies rely on experimental evidence to decide who to treat in a data-driven way \citep{hussam2022targeting,aiken2022machine}, even if \cite{haushofer2022targeting} pointed out the need for a more formalized approach to the targeting decision problem. The availability for the policymaker of appropriate tools to use the data in the decision process is probably necessary to guarantee a broader adoption of data-driven targeting strategies. Focusing on threshold policies, this paper explicitly formulates the decision problem, introduces implementable policies (the EWM and the SWM policy), and compares their asymptotic properties.

Turning to the threshold estimators, I already mentioned that the EWM policy exhibits the cube root of $n$ asymptotics studied by \cite{kim1990cube}, distinctive of several estimators in different contexts. Noteworthy examples are the maximum score estimator in choice models \citep{manski1975maximum}, the split point estimator in decision trees \citep{banerjee2007confidence}, and the risk minimizer in classification problems \citep{mohammadi2005asymptotics}, among others. Specific to my analysis is the emergence of the cube root asymptotic within a causal inference problem relying on the potential outcomes model, which is then mirrored in the regret's asymptotic distribution.

Addressing the cube root problem by smoothing the indicator in the objective function aligns closely with the strategy proposed by \cite{horowitz1992smoothed} for studying the asymptotic behavior of the maximum score estimator. Objective functions are nonetheless different, and in my context, I derive the asymptotic distribution for both the unsmoothed (EWM) and the smoothed (SWM) policies. This is convenient, as it allows me to compare not only the convergence rates but also the entire asymptotic distributions of the estimators and their regrets and study the asymptotic approximations in Monte Carlo simulations.

The rest of the paper is structured as follows. Section \ref{sec:overview} introduces the problem and outlines my analytical approach. Section \ref{sec:formal} derives formal results for the asymptotic distribution of the EWM and SWM policies and their regrets. In Section \ref{sec:mc}, I investigate finite sample performance of the EWM and SWM policies through Monte Carlo simulations, while in Section \ref{sec:empirical} I consider the analysis of experimental data from the National Job Training Partnership Act (JTPA) Study to compare the practical implications of the policies. Section \ref{sec:conclusion} concludes.

\section{Overview of the Problem} \label{sec:overview}

I consider the problem of a policymaker who wants to implement a binary treatment in a population of interest. Individuals are characterized by a vector of observable characteristics $\textbf{X} \in \R^d$, on which the policymaker bases the treatment assignment choice. A policy is hence a map $\pi(\textbf{x}): \R^d \rightarrow \{0,1\}$, from observable characteristics to the binary treatment status. The policymaker is utilitarian: its goal is to maximize the average welfare of the population. Indicating by $Y_1$ and $Y_0$ the potential outcomes with and without the treatment, population average welfare generated by a policy $\pi$ can be written as
\begin{gather}
    W(\pi) =  \E[Y_1 \pi(\textbf{X}) + Y_0 (1 - \pi(\textbf{X}) )].
\end{gather}
When treatment effects are heterogeneous, the same treatment can have opposite average effects across individuals with different $\textbf{X}$'s. For this reason, the policy assignment may vary with $\textbf{X}$: the policymaker wants to target only those who benefit from being treated, to maximize the average welfare.

The policy learning literature has considered several classes $\Pi$ of policy functions, such as linear eligibility indexes, decision trees, or monotone rules, discussed in \cite{kitagawa2018should,athey2021policy, mbakop2021model}. This paper focuses on threshold policies, a specific class of policy functions that can be represented as
\begin{gather}
    \pi(\textbf{X}) = \pi(X,t) = \mathbf{1} \{X > t \}.
\end{gather}
The threshold policy assigns the treatment whenever the scalar index $X \in \R$, one of the observable characteristics, exceeds a threshold $t$, the parameter to be chosen.

Threshold policies are widespread: they regulate, beyond others, organ transplants \cite{kamath2007model}, taxation \citep{taylor2003corporation}, and access to social welfare programs \citep{ card2008impact,crost2014aid}. Their key advantage seems to be simplicity: threshold policies are easy for eligible individuals to understand, simple for policymakers to implement and monitor, and transparent, with clearly defined eligibility criteria -- unlike more opaque black-box algorithms. These factors often justify the use of threshold policies even when a more structured alternative policy class may theoretically deliver higher welfare. In practice, these alternatives require additional resources for implementation, adoption, and monitoring -- potentially offsetting the welfare gains. Modeling this trade-off goes beyond the scope of this paper, where the restriction to the threshold policy class is taken as given and should not be interpreted as an endorsement of threshold policies. Nonetheless, it is worth noting that if the conditional average treatment effect (CATE) is monotone in $X$ and exhibits sign heterogeneity, as is often the case in applications, then the threshold policy is optimal among all policies that use only $X$ to assign the treatment.

I will focus on the case when the index $X$ is chosen before the experiment. Population welfare depends only on threshold $t$, and can be written as
\begin{gather}
    W(\pi) = W(t) = \E[Y_1\mathbf{1}\{X > t\} + Y_0\mathbf{1}\{X \leq t\}]. 
\end{gather}

Choosing the policy is equivalent to choosing the threshold. If the joint distribution of $Y_1$, $Y_0$, and $X$ were known, the policymaker would implement the policy with threshold $t^*$ defined as:
\begin{gather} \label{eq:estimand}
    t^* \in \argmax_{t} \E[Y_1\mathbf{1}\{X> t\} + Y_0\mathbf{1}\{X \leq t\}]
\end{gather}
which would guarantee the maximum achievable welfare $W(t^*)$.

The problem described in equation \eqref{eq:estimand} is infeasible since the joint distribution of $Y_1$, $Y_0$, and $X$ is unknown. The policymaker observes an experimental sample $Z=\{Z_i\}_{i=1}^n=\{Y_i,D_i,X_i\}$, where $Y$ is the outcome of interest, $D$ the randomly assigned treatment status, and $X$ the policy index. Experimental data, which allows to identify the conditional average treatment effect, are used to learn the threshold policy $\hat{t}_n = \hat{t}_n(Z)$, function of the observed sample.

Statistical decision theory deals with the problem of choosing the map $\hat{t}_n$. First, it is necessary to specify the decision problem the policymaker faces. For any threshold policy $\hat{t}_n$, define the regret $\mathcal{R}(\hat{t}_n)$:
\begin{gather}
    \mathcal{R}(\hat{t}_n) = W(t^*) - W(\hat{t}_n),
\end{gather}
a measure of welfare loss indicating the suboptimality of policy $\hat{t}_n$. The regret depends on the unknown data distribution: the policymaker specifies a state space, and searches for a policy that does well uniformly for all the data distributions in the state space. Following \cite{manski2004statistical}, statistical decision theory has mainly focused on the problem of minimizing the maximum expected regret, looking for a policy $\hat{t}_n$ that does uniformly well on average across repeated samples.

Directly solving the constrained minimization problem of the functional $\sup \E[\mathcal{R}(\hat{t}_n)]$ is impractical: the literature instead focuses on considering a specific policy map and studying its properties, for example showing its rate optimality, through finite sample valid \citep{kitagawa2018should} or asymptotic \citep{athey2021policy} arguments. Following this approach, I characterize and compare some properties for the regret of two different threshold policies, the Empirical Welfare Maximizer (EWM) policy, commonly studied in the literature, and the novel Smoothed Welfare Maximizer (SWM) policy.

\cite{kitagawa2018should} derive finite sample bounds for the expected regret of the EWM policy for a wide range of policy function classes. In their results, the policy class $\Pi$ affects the bounds only through its VC dimension, and the knowledge of $\Pi$ is not further exploited. Conversely, I leverage the additional structure from the knowledge of the policy class and characterize the asymptotic distribution of the regret for the EWM and the SWM threshold policies, comparing how their regrets depend on the data distribution. My results could hence be of interest also when decision problems not involving the expected regret are considered, as in \cite{manski2023statistical} and \cite{kitagawa2022treatment}: I characterize the asymptotic behavior of regret quantiles, and the asymptotic distributions can be used to simulate expectations of their non-linear functions.

To derive my results, I take advantage of the link between a threshold policy function $\hat{t}_n$ and its regret $\mathcal{R}(\hat{t}_n)$. Let $\{r_n\}$ be a sequence such that $r_n \rightarrow \infty$ for $n \rightarrow \infty$, and suppose that $r_n (\hat{t}_n - t^*)$ converges to a non degenerate limiting distribution, i.e $(\hat{t}_n - t^*) = O_p(r_n^{-1})$. 

Assume function $W(t)$ to be twice continuously differentiable, and consider its second-order Taylor expansion around $t^*$:
\begin{gather*}
    W(\hat{t}_n) = W(t^*) + \underbrace{W'(t^*)}_{=0} \left(\hat{t}_n-t^*\right) + \frac{1}{2} W''(\tilde{t}) \left(\hat{t}_n-t^*\right)^2
\end{gather*}
where $|\tilde{t}-t^*| \leq |\hat{t}_n - t^*|$, and $W'(t^*)=0$ by optimality of $t^*$. The previous equation can be written as
\begin{gather} \label{eq:asympt_regret_d}
    r_n^2 \mathcal{R}(\hat{t}_n) = \frac{1}{2} W''(\tilde{t}) \left[r_n\left(\hat{t}_n-t^*\right)\right]^2,
\end{gather}
establishing a relationship between the convergence rates of $\hat{t}_n$ and $\mathcal{R}(\hat{t}_n)$, and between their asymptotic distributions. Equation \eqref{eq:asympt_regret_d} therefore shows how the rate of convergence and the asymptotic distribution of regret $\mathcal{R}(\hat{t}_n)$ can be studied through the rate of convergence and the asymptotic distribution of policy $\hat{t}_n$. In the next section, I consider the EWM policy $\hat{t}^e_n$ and the SWM policy $\hat{t}^s_n$: through their asymptotic behaviors, I characterize the asymptotic distributions of their regrets $\mathcal{R}(\hat{t}^e_n)$ and $\mathcal{R}(\hat{t}^s_n)$.

\begin{rem}
{\normalfont (Ceteris Paribus Optimality)} Following the literature in statistical decision theory, I assume that the experiment is conducted in a population with the same distribution as the one where the policy will be implemented. This implicitly assumes that individuals in the target population do not change their behavior in response to the policy -- for instance, by altering their covariates to gain access to the treatment. In some contexts, this assumption may be unrealistic. Existing empirical studies on threshold policies highlight this issue: in regression discontinuity design, manipulation tests are specifically aimed at detecting such reactions to the policy. If manipulation occurs, it invalidates the optimality of the policy estimated in the experiment.
\end{rem}

\section{Formal Results} \label{sec:formal}
Let $Y_0$ and $Y_1$ be scalar potential outcomes, $D$ the binary treatment assignment in the experiment, and $X$ the observable index. $\{Y_0,Y_1, D, X \}$ are random variables distributed according to the distribution $P$. They satisfy the following assumptions, which guarantee the identification of the optimal threshold:

\begin{ass} \label{ass:identification}
{\normalfont (Identification)} 
\begin{enumerate}[label=1.\arabic*]
\item \label{ass:sutva}
    {\normalfont (No interference)} Observed outcome $Y$ is related with potential outcomes by the expression $Y= DY_1 + (1-D)Y_0$.
\item \label{ass:random}
    {\normalfont (Unconfoundedness)} Distribution $P$ satisfies $D \perp\!\!\!\perp (Y_0,Y_1) |X$.
\item \label{ass:overlap}
    {\normalfont (Overlap)} Propensity score $p(x)=\E[D|X=x]$ is assumed to be known and such that $p(x) \in (\eta,1-\eta)$, for some $\eta \in (0,0.5)$.
\item \label{ass:jointdist}
    {\normalfont (Joint distribution)} Potential outcomes $(Y_0, Y_1)$ and index $X$ are continuous random variables with joint probability density function $\varphi (y_0,y_1,x)$, and marginal densities $\varphi_0$, $\varphi_1$, and $f_x$ respectively. Expectations $\E[Y_0|x]$ and $\E[Y_1|x]$, for $x$ in the support of $X$, exist.
\end{enumerate}
\end{ass}

Assumptions \ref{ass:sutva}, \ref{ass:random}, and \ref{ass:overlap} are standard assumptions in many causal models. Assumption \ref{ass:sutva} requires the outcome of each unit to depend only on their treatment status, excluding spillover effects. Assumption \ref{ass:random} requires random assignment of the treatment, conditionally on $X$. Assumption \ref{ass:overlap} requires that, for any value of $X$, there is a positive probability of observing both treated and untreated units. Probabilities of being assigned to the treatment may vary with $X$, allowing for stratified experiments.

Assumption \ref{ass:jointdist} specifies the focus on continuous outcome and index. While it would be possible to accommodate discrete $Y_0$ and $Y_1$, maintaining the continuity of $X$ remains essential. The arguments developed in this paper, in fact, are not valid for a discrete index: my focus is on studying optimal threshold policies in contexts where the probability of observing any value on the support of the index $X$ is zero, and the threshold must be chosen from a continuum of possibilities. 

Under Assumption \ref{ass:identification}, optimal policy $t^*$ defined in \eqref{eq:estimand} can be written as
\begin{align}
    t^* \in & \argmax_{t} \E_P[Y_1\ind\{X> t\} + Y_0\ind\{X \leq t\}] \\
    = & \argmax_{t} \E_P\left[\left(\frac{D Y}{p(X)} - \frac{(1-D) Y}{(1-p(X))} \right) \ind\{X > t\}\right]
\end{align}
and is hence identified. This standard result specifies under which conditions an experiment allows to identify $t^*$.

\subsection{Empirical Welfare Maximizer Policy}

Policymaker observes an i.i.d. random sample $Z=\{Y_i,D_i,X_i\}$ of size $n$ from $P$, and considers the Empirical Welfare Maximizer policy $\hat{t}^e_n$, the sample analog of $t^*$ in equation \eqref{eq:estimand}\footnote{The objective function is piecewise constant, so solving the optimization problem requires evaluating the function $n+1$ times. The solution is the convex set of points in $\R$ that achieve the maximum of these values.}:
\begin{gather} \label{eq:estimator}
\hat{t}^e_n=\argmax_{t} \frac{1}{n} \sum_{i=1}^n \left(\frac{D_i Y_i}{p(X_i)}\ind\{X_i > t\} + \frac{(1-D_i) Y_i}{(1-p(X_i))}\ind\{X_i \leq t\} \right).
\end{gather}
Policy $\hat{t}^e_n$ can be seen as an extremum estimator, maximizer of a function not continuous in $t$.

\subsubsection{Consistency of $\hat{t}^e_n$}

First, I will prove that $\hat{t}^e_n$ consistently estimates the optimal threshold $t^*$, implying that $\mathcal{R}(\hat{t}^e_n) \rightarrow ^p 0$. To prove this result, I need the following assumptions on the data distribution.

\begin{ass} \label{ass:consistency}
{\normalfont (Consistency)} 
\begin{enumerate}[label=2.\arabic*]
\item \label{ass:uniq}
    {\normalfont (Maximizer $t^*$)} Maximizer $t^* \in \mathcal{T} $ of $\E[(Y_1-Y_0) \ind\{X> t\}]$ exists and is unique. It is an interior point of the compact parameter space $\mathcal{T} \subseteq \R$.
\item \label{ass:sqintegr}
    {\normalfont (Square integrability)} Conditional expectations $\E[Y_0^2|X]$ and $\E[Y_1^2|X]$ exist.
\item \label{ass:smoot}
    {\normalfont (Smoothness)} In a neighbourhood of $t^*$, density $f_x(x)$ is positive, and function $\E[(Y_1-Y_0) \ind\{X> t\}]$ is at least $s$-times continuously differentiable in $t$.
\end{enumerate}
\end{ass}

By requiring the existence of the optimal threshold in the interior of the parameter space, Assumption \ref{ass:uniq} is assuming heterogeneity in the sign of the conditional average treatment effect $\E[Y_1 - Y_0 | X]$. It is because of this heterogeneity that the policymaker implements the threshold policy, targeting groups that would benefit from being treated. The assumption neither excludes the multiplicity of local maxima, as long as the global one is unique, nor excludes unbounded support for $X$, but requires the parameter space to be compact. Uniqueness of the maximizer is not required in the standard analysis of the EWM policy, where partial identification of the optimal policy is allowed. A sufficient condition for Assumption \ref{ass:uniq}, easy to interpret and plausible in many applications, is that the conditional average treatment effect has negative and positive values, and crosses zero exactly once.

Assumption \ref{ass:sqintegr} requires that the conditional potential outcomes have finite second moments and is satisfied when $Y$ is assumed to be bounded (as in \cite{kitagawa2018should}).

Assumption \ref{ass:smoot} will be used with increasing values of $s$ to prove different results. To prove consistency, it needs to hold for $s=0$, requiring the continuity of the objective function $W(t)$ in a neighborhood of $t^*$. The derivative of $\E[(Y_1-Y_0) \ind\{X> t\}]$ with respect to $t$ is equal to $-f_x(t) \tau(t)$, where $\tau(x) = \E[Y_1-Y_0 | X=x]$ is the conditional average treatment effect. Assumption \ref{ass:smoot} with $s \geq 1$ hence requires smoothness of $f_x(x)$ and $\tau(x)$, in a neighborhood of $t^*$.

The following theorem proves the consistency of $\hat{t}^e_n$ for $t^*$.

\begin{thm} \label{thm:con}
Consider the EWM policy $\hat{t}^e_n$ defined in equation \eqref{eq:estimator} and the optimal policy $t^*$ defined in equation \eqref{eq:estimand}. Under Assumptions \ref{ass:identification} and \ref{ass:consistency} (with $s=0$),
\begin{gather*}
    \hat{t}^e_n \rightarrow^{a.s.} t^*
\end{gather*}
i.e. $\hat{t}^e_n$ is a strongly consistent estimator for $t^*$.
\end{thm}

\subsubsection{Asymptotic Distribution for $\hat{t}^e_n$}

The fact that $\hat{t}^e_n$ is the maximizer of a function not continuous in $t$ directly affects the convergence rate and the asymptotic distribution. The EWM policy $\hat{t}^e_n$ exhibits the ``cube root asymptotics'' behavior studied in \cite{kim1990cube}, the same as, beyond others, the maximum score estimator \citep{manski1975maximum}, and the split point estimator in decision trees \citep{banerjee2007confidence}.

The limiting distribution is not Gaussian, and its derivation requires two additional regularity conditions on $P$:

\begin{ass} \label{ass:asymptotic}
{\normalfont (Asymptotic Distribution)} 
\begin{enumerate}[label=3.\arabic*]
\item \label{ass:nonflattau}
    {\normalfont (Non-flat $\tau(X)$)} In a neighbourhood of $t^*$, the derivative of the conditional average treatment effect, $\frac{\partial \E\left[Y_1 - Y_0 | X \right]}{\partial X}$, is non-zero.
\item \label{ass:tail}
     {\normalfont (Tail condition)} Let $\varphi_1$ and $\varphi_0$ be the probability density functions of $Y_1$ and $Y_0$. Assume that, as $|y| \rightarrow \infty$, $\varphi_1(y) = o(|y|^{-(4+\delta)})$ and $\varphi_0(y) = o(|y|^{-(4+\delta)})$, for $\delta > 0$.
\end{enumerate}
\end{ass}

Assumption \ref{ass:nonflattau} requires that, close to the maximizer $t^*$, the conditional average treatment effect function $\tau(X)$ is not flat. If $\tau(X)$ were flat in the neighborhood of $t^*$, it would be harder for the estimator to find the exact maximizer, leading to a slower rate of convergence. The excluded flat $\tau(X)$ corresponds to a situation where estimating the threshold precisely is less critical, as the CATE remains zero even in a neighborhood of the optimal threshold.

Assumption \ref{ass:tail} requires that the tails of the distributions of the potential outcomes are not too fat. It is generally satisfied by any bounded distribution, and by distributions in the exponential family, while is violated, for example, by the Student's t-distribution with fewer than four degrees of freedom.

The following theorem gives the asymptotic distribution of $\hat{t}^e_n$.

\begin{thm} \label{thm:asyd}
Consider the EWM policy $\hat{t}^e_n$ defined in equation \eqref{eq:estimator} and the optimal policy $t^*$ defined in equation \eqref{eq:estimand}. Under Assumptions \ref{ass:identification}, \ref{ass:consistency} (with $s=2$), and \ref{ass:asymptotic}, as $n\rightarrow \infty$,
\begin{gather}
    n^{1 / 3}\left(\hat{t}^e_n-t^*\right) \rightarrow^d (2\sqrt{K}/H)^{\frac{2}{3}}\argmax_r \left(B(r) - r^2 \right)
\end{gather}
where $B(r)$ is the two-sided standard Brownian motion process, and $K$ and $H$ are
\begin{align*}
    K
    =& f_x(t^*) \left(\frac{1}{p(t^*)} \E[Y_1^2|X = t^*] + \frac{1}{1-p(t^*)} \E[Y_0^2|X = t^*] \right) \\
    H =&  f_x(t^*) \left(\frac{\partial \E\left[Y_1 - Y_0 | X = t^* \right]}{\partial X} \right).
\end{align*}
\end{thm}

The limiting distribution of $n^{1 / 3}\left(\hat{t}^e_n-t^*\right)$ is of Chernoff type \citep{chernoff1964estimation}. The Chernoff's distribution is the probability distribution of the random variable $\argmax_r B(r) - r^2$, where $B(r)$ is the two-sided standard Brownian motion process. The process $B(r) - r^2$ can be simulated, and the distribution of $\argmax_r B(r) - r^2$ numerically studied. \cite{groeneboom2001computing} report values for selected quantiles.

It's worth noticing how the variance of $\hat{t}^e_n$ depends on the data distribution. $K$ and $H$ are functions of the density of $X$, the variance of the potential outcomes, and the derivative of the CATE at $t^*$. The optimal threshold is estimated with more precision when more data around the optimal threshold are available (larger density), when the treatment effect changes more rapidly (larger derivative of CATE), and when the outcomes have less variability.

Results in Theorem \ref{thm:asyd} can be used to derive asymptotic valid confidence intervals for $\hat{t}^e_n$, as discussed in Appendix \ref{app:ci}. More interestingly, they can be combined with Equation \ref{eq:asympt_regret_d} to characterize the asymptotic distribution of the regret $\mathcal{R}(\hat{t}^e_n)$, as derived in the following corollary.

\begin{corollary} \label{cor:ewm_regret}
    The asymptotic distribution of regret $\mathcal{R}(\hat{t}^e_n)$ is:
    \begin{gather*}
        n^{\frac{2}{3}} \mathcal{R}(\hat{t}^e_n) \rightarrow^d  \left( \frac{2K^{2}}{H} \right)^\frac{1}{3}  \left( \argmax_r B(r) - r^2 \right)^2.
    \end{gather*}
    The expected value of the asymptotic distribution is $K^\frac{2}{3} H^{-\frac{1}{3}} C^e$, where $$C^e= \sqrt[3]{2} \E\left[\left( \argmax_r B(r) - r^2 \right)^2\right]$$ is a constant not dependent on $P$.
\end{corollary}

For the regret of the EWM policy, Corollary \ref{cor:ewm_regret} establishes a $n^{\frac{2}{3}}$ rate, faster than the $\sqrt{n}$ rate found to be the optimal for the EWM expected regret \citep{kitagawa2018should}. It is essential to highlight the differences between the two results: Corollary \ref{cor:ewm_regret} is about pointwise convergence in distribution, while the main results by \cite{kitagawa2018should} establish a uniform rate for the expected regret. The family of distributions they consider may violate the assumptions in Theorem \ref{thm:asyd}, for example including cases where the CATE is flat at the optimal threshold. In contrast, Corollary \ref{cor:ewm_regret} is derived for distributions that satisfy Assumptions \ref{ass:identification}, \ref{ass:consistency} and \ref{ass:asymptotic}, which imply $H \neq 0$.

\cite{kitagawa2018should} also discuss how additional assumptions on the distribution $P$ can lead to a faster convergence rate for the regret. They show that if a certain margin condition on the data distribution holds, the convergence rate can improve. When the CATE function $\tau(X)$ is non-flat, the regret achieves a uniform convergence rate of $n^{\frac{2}{3}}$. Although the two rates are the same and rely on a similar condition on the CATE, they are derived from different, non-nested sets of assumptions. For instance, \cite{kitagawa2018should} assume that the policy class is correctly specified and $Y$ is bounded, but don't require the CATE to be smooth. The fact that my pointwise convergence rate coincides with their uniform rate when all assumptions are met suggests that my rate for the EWM threshold policy is also uniformly optimal.

The result in Corollary \ref{cor:ewm_regret} does not imply convergence in the mean, and the expected value of the asymptotic distribution is presented as a summary statistic -- a measure of the location of the asymptotic distribution.

\subsection{Smoothed Welfare Maximizer Policy}

Corollary \ref{cor:ewm_regret} shows how the cube root of $n$ convergence rate of the EWM policy directly impacts the convergence rate of its regret. In this section, I propose an alternative threshold policy, the Smoothed Welfare Maximizer policy, that achieves a faster rate of convergence and hence guarantees a faster rate of convergence for its regret. My approach exploits some additional smoothness assumptions on the distribution $P$: Corollary \ref{cor:ewm_regret} holds when $f_x(x)$ and $\tau(x)$ are assumed to be at least once differentiable; if they are at least twice differentiable, the SWM policy guarantees a $n^\frac{4}{5}$ convergence rate for the regret. Note that asking the density of the index and the conditional average treatment effect to be twice differentiable seems plausible for many applications. In the context of policy learning, it is, for example, assumed by \cite{athey2021policy} to derive their results\footnote{This assumption is implied by the high-level assumptions made in the paper, as the authors discuss in footnote 15.}.

My approach involves smoothing the objective function in \eqref{eq:estimator}, in the same spirit as the smoothed maximum score estimator proposed by \cite{horowitz1992smoothed} to deal with inference for the maximum score estimator \citep{manski1975maximum}. The Smoothed Welfare Maximizer (SWM) policy $\hat{t}^s_n$ is defined as\footnote{The objective function is continuous and smooth in $t$, which allows for the use of gradient descent algorithms. However, the function is not concave and may have many local maxima. Annealing algorithms can be employed to find the global maximum.}:
\begin{gather} \label{eq:estimator_sm}
\hat{t}^s_n=\argmax_{t} \underbrace{\frac{1}{n} \sum_{i=1}^n \left[\left(\frac{D_i Y_i}{p(X_i)} - \frac{(1-D_i) Y_i}{(1-p(X_i))} \right) k\left(\frac{X_i - t}{\sigma_n} \right) \right]}_{ \hat{S}_n(t,\sigma_n)}
\end{gather}
where $\sigma_n$ is a sequence of positive real numbers such that $\lim_{n \rightarrow \infty} \sigma_n = 0$, and the function $k(\cdot)$ satisfies: 
\begin{ass} \label{ass:kernel}
     {\normalfont (Kernel function)} Kernel function $k(\cdot): \R \rightarrow \R$ is continuous, bounded, and with limits $\lim_{x \rightarrow - \infty} k(x) = 0$ and $\lim_{x \rightarrow \infty} k(x) = 1$.
\end{ass}

In practice, the indicator function found in $\hat{t}^e_n$ is here substituted by a smooth function $k(\cdot)$ with the same limiting behavior, which guarantees the differentiability of the objective function. The bandwidth $\sigma_n$, decreasing with the sample size, ensures that when $n \to \infty$ the policy converges to the optimal one, as proved in the next section.

\subsubsection{Consistency of $\hat{t}^s_n$}

I start showing consistency of $\hat{t}^s_n$ for $t^*$, which implies $\mathcal{R}(\hat{t}^s_n) \rightarrow^p 0$.

\begin{thm} \label{thm:cons_sm}
Consider the SWM policy $\hat{t}^s_n$ defined in equation \eqref{eq:estimator_sm} and the optimal policy $t^*$ defined in equation \eqref{eq:estimand}. Under Assumptions \ref{ass:identification}, \ref{ass:consistency} (with $s=0$), and \ref{ass:kernel}, as $n \rightarrow \infty$,
\begin{gather*}
    \hat{t}^s_n \rightarrow^{a.s.} t^*
\end{gather*}
i.e. $\hat{t}^s_n$ is a strongly consistent estimator for $t^*$.
\end{thm}

Theorems \ref{thm:cons_sm} and \ref{thm:con} are analogous: they rely on the same assumptions on the data (Assumptions \ref{ass:identification} and \ref{ass:consistency}) to prove the consistency of $\hat{t}^e_n$ and $\hat{t}^s_n$. Where the two policies differ is in the asymptotic distributions: smoothness in the objective function for $\hat{t}^s_n$ guarantees asymptotic normality, but also introduces a bias, since the bandwidth $\sigma_n$ equals zero only in the limit, which emerges in the limiting distribution.

\subsubsection{Asymptotic Distribution for $\hat{t}^s_n$}

Deriving this asymptotic behavior of $\hat{t}^s_n$ requires an additional assumption on the rate of bandwidth $\sigma_n$ and the kernel function $k$. Since both are chosen by the policymaker, the assumption is not a restriction on the data but a condition on properly picking $\sigma_n$ and $k$.

\begin{ass} \label{ass:asymptotic_sm}
{\normalfont (Bandwidth and kernel)} 
\begin{enumerate}[label=5.\arabic*]
\item \label{ass:rateofs}
    {\normalfont (Rate of $\sigma_n$)} $\frac{\log n}{n \sigma_n^4} \rightarrow 0$ as $n\rightarrow \infty$.
\item \label{ass:extraonk}
   {\normalfont (Kernel function)} Kernel function $k(\cdot): \R \rightarrow \R$ satisfies Assumption \ref{ass:kernel} and the following:
     \begin{itemize}
         \item $k(\cdot)$ is twice differentiable, with uniformly bounded derivatives $k'$ and $k''$.
         \item $\int k'(x)^4 dx$, $\int k''(x)^2 dx$, and $\int |x^2 k''(x)| dx$ are finite.
         \item For some integer $h \geq 2$ and each integer $i\in[1,h]$, $\int |x^i k'(x)| dx =0$ for \\ $i<h$ and $\int |x^h k'(x)| dx =d \neq 0$, with $d$ finite.
         \item For any integer $i \in [0,h]$, any $\eta >0$, and any sequence $\sigma_n \rightarrow 0$, \\ $\lim_{n\rightarrow \infty} \sigma_n^{i-h} \int_{|\sigma_n x| > \eta} | x^i k'(x)| dx = 0$, and $\lim_{n\rightarrow \infty} \sigma_n^{-1} \int_{|\sigma_n x| > \eta} | k''(x)| dx = 0$.
         \item $\int x k''(x) dx = 1$, $\lim_{n\rightarrow \infty}  \int_{|\sigma_n x| > \eta} | x k''(x)| dx = 0$.
     \end{itemize} 
\end{enumerate}
\end{ass}

An example of a function $k$ satisfying Assumption \ref{ass:extraonk} with $h=2$ is the cumulative distribution function of the standard normal distribution.

I can now derive the asymptotic distribution of $\hat{t}^s_n$.
\begin{thm} \label{thm:asynormsmooth}
    Consider the SWM policy $\hat{t}^s_n$ defined in equation \eqref{eq:estimator_sm} and the optimal policy $t^*$ defined in equation \eqref{eq:estimand}. Under Assumptions \ref{ass:identification}, \ref{ass:consistency} (with $s=h + 1$ for some $h\geq 2$), \ref{ass:nonflattau}, and \ref{ass:asymptotic_sm}, as $n \rightarrow \infty$:
    \begin{enumerate}
        \item if $n \sigma_n^{2h + 1} \rightarrow \infty$, $$\sigma_n^{-h}(\hat{t}^s_n - t^*) \rightarrow^p H^{-1}A;$$
        \item  if $n \sigma_n^{2h + 1} \rightarrow \lambda < \infty$, $$(n\sigma_n)^{\frac{1}{2}}(\hat{t}^s_n - t^*) \rightarrow^d \mathcal{N}(\lambda^{\frac{1}{2}}H^{-1}A, H^{-2}\alpha_2 K);$$
    \end{enumerate}
    where $A$, $\alpha_1$, and $\alpha_2$ are:
    \begin{align}
    A =& -\frac{1}{h!} \alpha_1 \int_y \left(Y_1 - Y_0 \right) \varphi^{h}_x(y,t^*) d y \\
    \alpha_1 =& \int_\zeta \zeta^h k'\left(\zeta \right) d \zeta \\
    \alpha_2 =& \int_\zeta k'\left(\zeta \right)^2 d \zeta.
\end{align}
\end{thm}

The asymptotic distribution of $(n\sigma_n)^{\frac{1}{2}}(\hat{t}^s_n - t^*)$ is normal, centered at the asymptotic bias $\lambda^{\frac{1}{2}}H^{-1}A$ introduced by the smoothing function, which exploits local information giving non-zero weights to treated units in the untreated region and vice versa. The bias and the variance of the distribution depend on the population distribution through $K$ and $H$, as for the EWM policy, and also through $A$, a new term that determines the bias. In the definition of $A$, $\varphi^{h}_x$ is the $h$ derivative with respect to $x$ of $\varphi(y_0,y_1,x)$, the joint density distribution of $Y_0$, $Y_1$, and $X$: the integral in the expression for $A$ is the $h$-derivative of $f_x(X) \tau(X)$ computed in $X=t^*$, whose existence is guaranteed by Assumption \ref{ass:smoot} with $s=h+1$. $\alpha_1$ and $\alpha_2$ depends only on kernel function $k$, and are hence known.

Theorems \ref{thm:asyd} and \ref{thm:asynormsmooth} differ in the smoothness requirements imposed by Assumption \ref{ass:smoot}. Theorem \ref{thm:asyd} requires the function $\E[(Y_1-Y_0) \ind\{X> t\}]$ to be at least twice differentiable, while Theorem \ref{thm:asynormsmooth} requires three derivatives. The additional smoothness condition is necessary because certain steps in the proof rely on a Taylor expansion of the joint distribution of $Y$ and $X$ which requires $s \geq 3$ to exist. When $s=2$, the asymptotic distribution derived in Theorem \ref{thm:asynormsmooth} for $\frac{\partial \hat{S}_n(\hat{t}^s_n,\sigma_n)}{\partial^2 t} (n\sigma_n)^{\frac{1}{2}}(\hat{t}^s_n - t^*)$ no longer holds, as $\frac{\partial \hat{S}_n(\hat{t}^s_n,\sigma_n)}{\partial^2 t}$, the second derivative of the objective function in the SWM policy definition (Equation  \ref{eq:estimator_sm}), may not have a bounded limiting distribution. However, if such a bounded limiting distribution, denoted by $\tilde{H}$, exists, then $(n\sigma_n)^{\frac{1}{2}}(\hat{t}^s_n - t^*)$ remains $O_p(1)$ even when $s=2$, with an asymptotic distribution that depends on the unknown term $\tilde{H}$ rather than on $H$, as in Theorem \ref{thm:asynormsmooth}.

As for the EWM policy, results in Theorem \ref{thm:asynormsmooth} can be used to derive asymptotic valid confidence intervals for $\hat{t}^s_n$ (see Appendix \ref{app:ci}), and, combined with equation \ref{eq:asympt_regret_d}, to characterize the asymptotic distribution of the regret $\mathcal{R}(\hat{t}^s_n)$, as derived in the next corollary.

\begin{corollary} \label{cor:smoot_regret}
    Asymptotic distribution of regret $\mathcal{R}(\hat{t}^s_n)$ is:
    \begin{gather*}
        n \sigma_n \mathcal{R}(\hat{t}^s_n) \rightarrow^d  \frac{1}{2} \frac{\alpha_2 K}{H} \chi^2\left(1,\frac{\lambda A^2}{\alpha_2 K}\right)
    \end{gather*}
    where $\chi^2\left(1,\frac{\lambda A^2}{\alpha_2 K}\right)$ is a non-centered chi-squared distribution with 1 degree of freedom and non-central parameter $\frac{\lambda A^2}{\alpha_2 K}$.
    The expected value of the asymptotic distribution is:
    \begin{align} \label{eq:expregretswm}
        \frac{1}{2} \frac{\alpha_2 K}{H} \left(1 + \frac{\lambda A^2}{\alpha_2 K} \right) =  \frac{\alpha_2}{2} \frac{ K}{H} + \frac{1}{2} \frac{\lambda A^2}{H}.
    \end{align}
    Let $\sigma_n = (\lambda/n)^{1/(2h +1)}$ with $\lambda \in (0, \infty)$. The expectation of the asymptotic regret is minimized by setting $\lambda = \lambda^* = \frac{\alpha_2 K}{2hA^2 }$: in this case, the expectation of the asymptotic distribution scaled by $n^\frac{2h}{2h+1}$ is $ A^{\frac{2}{2h+1}} K^{\frac{2h}{2h+1}} H^{-1} C^s$, where $C^s = \frac{2h+1}{2} \left( \frac{\alpha_2}{2h} \right)^\frac{2h}{2h+1}$ is a constant not dependent on $P$.
\end{corollary}

With the optimal bandwidth $\sigma_n =O_p(n^{-\frac{1}{2h + 1}})$, the regret converges at $n^{\frac{2h}{2h + 1}}$ rate. For $h \geq 2$, this implies that the regret converges faster with the SWM than with the EWM policy: the extra smoothness assumption has been exploited to achieve a better rate for the asymptotic regret. When $h=1$, if additional regularity assumptions ensure the existence of $\tilde{H}$, the SWM policy attains the same convergence rate as the EWM policy. As for Corollary \ref{cor:ewm_regret}, the result in Corollary \ref{cor:smoot_regret} does not imply convergence in the mean, and the expected value of the asymptotic distribution is reported as a measure of the location of the asymptotic distribution.

The comparison between Corollaries \ref{cor:ewm_regret} and \ref{cor:smoot_regret}, which characterize the asymptotic distributions of regrets $\mathcal{R}(\hat{t}^e_n)$ and $\mathcal{R}(\hat{t}^s_n)$, highlights how the data distribution $P$ differently influences the asymptotic behavior of the regrets through $H$, $K$, which affect both distributions but with different exponents, and $A$, the bias term that affects only the SWM policy. The relevance of these asymptotic results relies on their ability to approximate behaviors in finite samples. After all, policymakers only have access to finite experimental data. Building on the theoretical results derived above, the next section uses Monte Carlo simulations to analyze the finite-sample regrets associated with the EWM and SWM policies.

\section{Monte Carlo Simulations} \label{sec:mc}

I examine the finite sample properties of the EWM and SWM policies using Monte Carlo simulations. The scope of this section is twofold: first, I will provide examples of data generating processes that lead to different rankings for the two policies in terms of median asymptotic regret, to illustrate how in the asymptotic results in Corollaries \ref{cor:ewm_regret} and \ref{cor:smoot_regret} the convergence rate and the limiting distribution interact. Then, I will verify how the asymptotic results approximate the finite sample distributions, and compare the finite sample regrets of the two policies.

As data generating process, consider the following distribution $P$ of $(Y_0, Y_1, D, X)$:
\begin{align*}
    X &\sim \mathcal{N} (0,1) \\
    \epsilon_1 &\sim \mathcal{N}(0,\gamma) \\
    Y_1 &= X^3 + \beta_2 X^2 + \beta_1 X + \epsilon_1 \\
    Y_0 &\sim \mathcal{N}(0,\gamma) \\
    D &\sim \text{Bern}(p).
\end{align*}
Under $P$, the potential outcome $Y_0$ does not depend on the index $X$, and the treatment is randomly assigned with constant probability $p$. Parameter values are chosen such that $\E[Y_1|X=x]$ and hence $\E[Y_1 - Y_0|X=x]$ are increasing function of $x$, and the optimal threshold $t^*$ is 0. It can be verified that such $P$ implies the following:
\begin{align*}
    K =& \phi(0) \left(\frac{\gamma^2}{p} + \frac{\gamma^2}{1-p}\right) \\
    H =& \phi(0) \beta_1 \\
    A =& -\phi(0) \beta_2 \\
    W(t) =& \beta_2 \left( 1- \Phi(t) + t \phi(t) \right) +
    \beta_1 \phi(t) + \left( t^2 \phi(t) + 2\phi(t) \right)
\end{align*}
where $\phi(t)$ and $\Phi(t)$ are the probability density function and the cumulative density function of the standard normal distribution, respectively.

I consider two models characterized by different parameter values, reported in Table \ref{tab:param}:
\begin{table}[h] \centering \small
\begin{tabular}{c|cccc}
\hline \hline
 Model & $\gamma$ & $\beta_1$ & $\beta_2$ & $p$  \\ \hline
 1 & 1 & 1 & -0.5 & 0.5 \\
 2 & 3 & 0.5 & -1 & 0.5 \\ \hline
\end{tabular}
\caption{\small Parameters values for the simulation models.} \label{tab:param}
\end{table}

For the SWM policy, the kernel function is the cumulative distribution function of the standard normal distribution, which satisfies Assumption \ref{ass:extraonk} with $h=2$. Consequently, all analyses are conducted under Assumption \ref{ass:consistency} with $s=h+1=3$. Asymptotic regrets are computed with the infeasible optimal bandwidth $\sigma_n^*$. Table \ref{tab:asym_values} presents the values of $K$, $H$, and $A$, along with the medians of the asymptotic regret for both policies for sample sizes $n \in \{500, 1000, 2000, 3000\}$. Compared to Model 1, Model 2 entails larger $K$ and $A$, and smaller $H$, which lead to higher median regrets under both policies. Consider the case with $n=500$. In model 1, the median of the asymptotic regret is higher with the EWM policy, whereas in model 2, it's higher with the SWM policy: this confirms that the ranking of the asymptotic median regrets depends on the unknown data distribution $P$. Because of the fastest rate, though, as $n$ increases the SWM policy exhibits relatively better performance. Regardless of the specific distribution $P$, there exists a certain sample size beyond which the asymptotic median regret with the SWM policy becomes smaller. In model 2, when $n=1,000$, the inversion of ranking already occurs.

\begin{table}[!htbp] \centering \small
\begin{tabular}{@{\extracolsep{5pt}} ccccccc} 
\\[-1.8ex]\hline 
\hline \\[-1.8ex] 
 Model & n & EWM & SWM & K & H & A \\ 
\hline \\[-1.8ex] 
\multirow{4}{*}{1} & $500$ & $45.347$ & $18.459$ & $1.596$ & $0.399$ & $0.199$ \\ 
& $1,000$ & $28.567$ & $10.602$ & $1.596$ & $0.399$ & $0.199$ \\ 
& $2,000$ & $17.996$ & $6.089$ & $1.596$ & $0.399$ & $0.199$ \\ 
& $3,000$ & $13.733$ & $4.402$ & $1.596$ & $0.399$ & $0.199$ \\ 
\hline
\multirow{4}{*}{2} & $500$ & $188.651$ & $204.255$ & $9.575$ & $0.199$ & $0.399$ \\ 
& $1,000$ & $118.843$ & $117.314$ & $9.575$ & $0.199$ & $0.399$ \\ 
& $2,000$ & $74.866$ & $67.379$ & $9.575$ & $0.199$ & $0.399$ \\ 
& $3,000$ & $57.134$ & $48.714$ & $9.575$ & $0.199$ & $0.399$ \\ 
\hline \\[-1.8ex] 
\end{tabular}
  \caption{ \small Values of $K$, $H$, and $A$, and asymptotic median regrets using both EWM and SWM policies across different models, are presented. The median regret for the SWM policy is computed using the optimal bandwidth $\sigma_n^*$. To facilitate the reading, asymptotic median regrets have been scaled by a factor of 10,000.} 
  \label{tab:asym_values} 
\end{table}

To investigate the finite sample distributions of the regret, I draw samples of size $n$ from $P$ 5,000 times for each model. Each sample is used to estimate the thresholds $\hat{t}^e_n$ and $\hat{t}^s_n$. Estimating $\hat{t}^s_n$ requires specifying a bandwidth $\sigma_n$, for which I adopt the following method: I use the estimated policy $\hat{t}^e_n$ to compute $\hat{A}_{n}$ and $\hat{K}_{n}$, which are then used to compute the optimal $\hat{\lambda}_n^*$, and the optimal bandwidth $\hat{\sigma}_n^*$. In Appendix \ref{app:ci}, I provide and discuss formulas for estimators $\hat{A}_{n}$ and $\hat{K}_{n}$. The SWM policy $\hat{t}^s_{n}$ is hence estimated with this bandwidth $\hat{\sigma}_n^*$, which consistently estimates the optimal bandwidth $\sigma_n^*$ if $\hat{A}_{n}$ and $\hat{K}_{n}$ consistently estimate $A$ and $K$. I also consider $\hat{t}^s_{n}$ with the infeasible optimal $\sigma_n^*$ computed from the data generating process.

Estimates for $\hat{t}^e_n$ and $\hat{t}^s_n$ are used to compute regrets $\mathcal{R}(\hat{t}^e_n)$ and $\mathcal{R}(\hat{t}^s_n)$. I thus obtain the finite sample distributions of the regret, which can be compared with the asymptotic distributions derived in Corollaries \ref{cor:ewm_regret} and \ref{cor:smoot_regret}. Table \ref{tab:median_regret} presents the median of these finite sample and asymptotic distributions, also depicted in Figures \ref{fig:median_reg_model1} and \ref{fig:median_reg_model2}. Corresponding tables and figures for the mean regret are provided in Appendix \ref{app:mc}.

The last column of each table reports the ratio between the finite sample median regrets for the EWM and SWM policy, facilitating the comparison: a ratio larger than one indicates that the SWM policy outperforms the EWM policy. These ratios increase with the sample size, reflecting the faster asymptotic convergence rate of the SWM policy. Similar to the asymptotic results, in finite sample the SWM policy does relatively better as the sample size increases.

\begin{table}[!htbp] \centering \small
\begin{tabular}{@{\extracolsep{5pt}} cccccccc} 
\\[-1.8ex]\hline 
\hline \\[-1.8ex] 
 Model & n & \multicolumn{2}{c}{EWM} & \multicolumn{3}{c}{SWM} & Ratio  \\
\hline
& & empirical & asymptotic & empirical $\sigma_n^*$ & empirical $\hat{\sigma}_n^*$ & asymptotic \\
\hline \\[-1.8ex] 
\multirow{4}{*}{1} & $500$ & $39.809$ & $45.347$ & $15.113$ & $24.169$ & $18.459$ & $1.647$ \\ 
& $1,000$ & $27.547$ & $28.567$ & $8.692$ & $14.145$ & $10.602$ & $1.948$ \\ 
& $2,000$ & $18.816$ & $17.996$ & $5.073$ & $8.507$ & $6.089$ & $2.212$ \\ 
& $3,000$ & $14.043$ & $13.733$ & $3.471$ & $5.491$ & $4.402$ & $2.558$ \\ 
\hline
\multirow{4}{*}{2} & $500$ & $147.029$ & $188.651$ & $116.784$ & $159.876$ & $204.255$ & $0.920$ \\ 
& $1,000$ & $110.372$ & $118.843$ & $82.473$ & $118.695$ & $117.314$ & $0.930$ \\ 
& $2,000$ & $87.325$ & $74.866$ & $55.731$ & $79.436$ & $67.379$ & $1.099$ \\ 
& $3,000$ & $72.122$ & $57.134$ & $46.007$ & $68.656$ & $48.714$ & $1.050$ \\ 
\hline \\[-1.8ex] 
\end{tabular}
  \caption{ \small Finite sample and asymptotic median regret with the EWM and SWM policies across different models. Finite sample (empirical) values are computed through 5,000 Monte Carlo simulations. The last column reports the ratio between finite sample median regrets with the EWM and SWM policies.} 
  \label{tab:median_regret} 
\end{table}

\begin{figure}
\begin{center}
\resizebox{0.8\textwidth}{!}{
\includegraphics{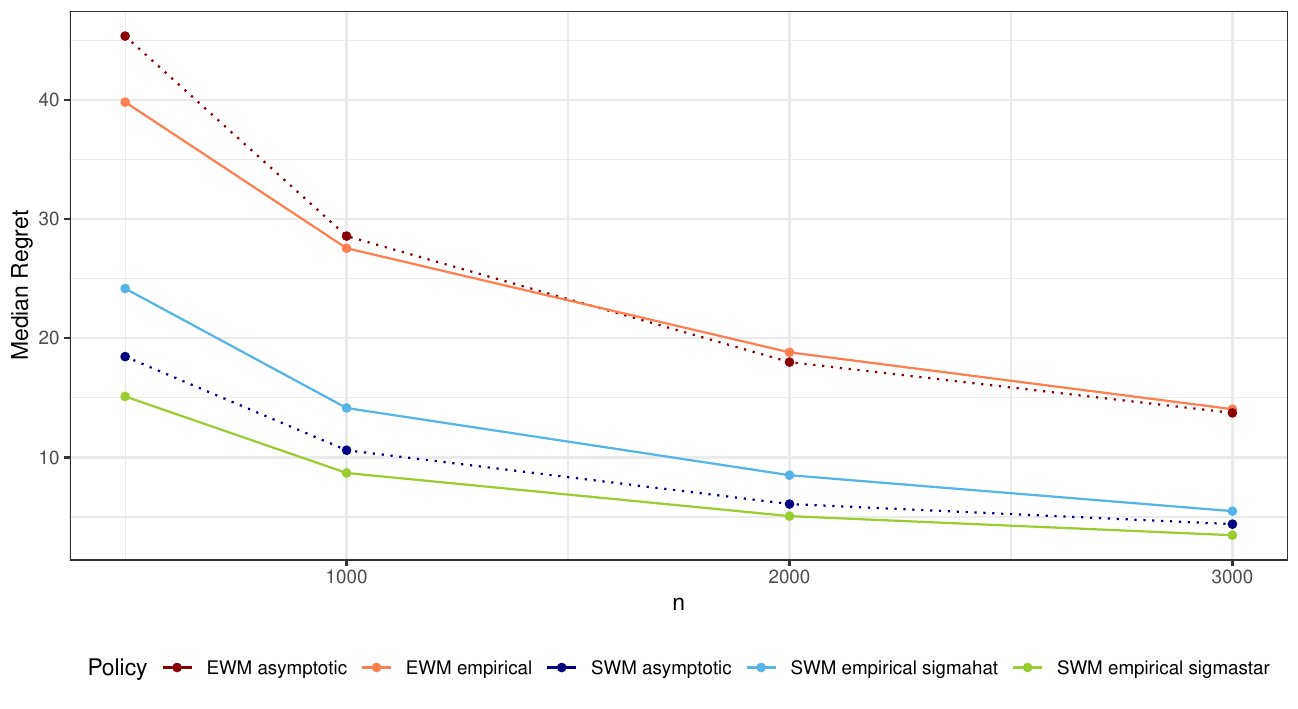}
}
\caption{\small The figure illustrates asymptotic and finite sample median regrets for the EWM and SWM policies in Model 1, corresponding to the values reported in Table \ref{tab:median_regret}.}
\label{fig:median_reg_model1}
\end{center}
\end{figure}

\begin{figure}
\begin{center}
\resizebox{0.8\textwidth}{!}{
\includegraphics{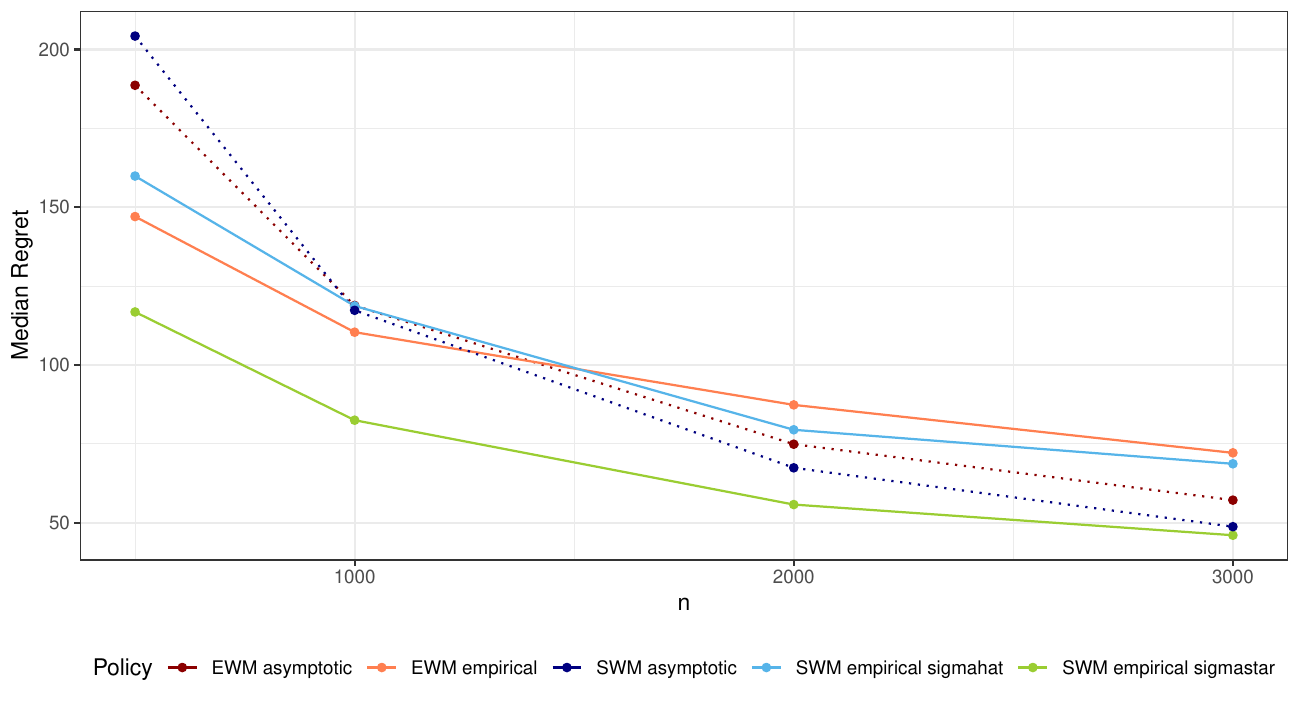}
}
\caption{\small The figure illustrates asymptotic and finite sample median regrets for the EWM and SWM policies in Model 2, corresponding to the values reported in Table \ref{tab:median_regret}.}
\label{fig:median_reg_model2}
\end{center}
\end{figure}

Simulations enable comparison between finite sample regrets and their asymptotic counterparts. Across all models and sample sizes, the asymptotic approximation for the feasible SWM policy (with the estimated bandwidth $\hat{\sigma}_n^*$) is relatively less accurate. Simulations suggest that this is partly attributable to the need for estimating an additional tuning parameter, the bandwidth $\sigma_n$. When the SWM policy is estimated using the infeasible optimal bandwidth $\sigma^*_n$, in fact, the asymptotic approximation is more accurate and the regret is smaller.

In Model 1, as illustrated in Figure \ref{fig:median_reg_model1}, both the finite sample and the asymptotic median regrets are lower for the SWM policy. Conversely, in Model 2 (illustrated in Figure \ref{fig:median_reg_model2}), the finite sample median regret is lower with the EWM policy. This confirms the impossibility of ranking the policies in a pointwise sense: different distributions $P$ result in different rankings for the finite sample median regrets.

It is important to note that the ranking of the EWM and the SWM policies indicated by the asymptotic results may differ from the actual finite sample comparison. Consider, for example, Model 2 with $n=1,000$: despite the asymptotic analysis suggesting a smaller median regret with the SWM policy, the EWM guarantees a smaller regret. In this scenario, even if $P$ were known in advance, choosing according to the asymptotic approximation would not have been optimal. As $n$ increases, the approximation improves, and the rankings based on asymptotic analysis and finite sample comparisons coincide.

Monte Carlo simulations have confirmed that the asymptotic results can approximate some finite sample behavior of the regrets, highlighting some caveats to consider when applying conclusions from asymptotic analysis to finite sample regrets with the EWM and SWM policies. However, they have not yet provided insight into the practical significance of the differences between the two policies, whether these differences are relevant or negligible in real-world scenarios. An empirical illustration is useful to answer these questions, illustrating the different implications that the policies may have.

\section{Empirical Illustration} \label{sec:empirical}

I consider the same empirical setting as \cite{kitagawa2018should}: experimental data from the National Job Training Partnership Act (JTPA) Study. \cite{bloom1997benefits} describes the experiment in detail. The study randomized whether applicants would be eligible to receive a mix of training, job-search assistance, and other services provided by the JTPA for a period of 18 months. Background information on the applicants was collected before treatment assignment, alongside administrative and survey data on the applicants' earnings over the subsequent 30 months.

I consider the same sample of 9,223 observations as in \cite{kitagawa2018should}. The treatment variable $D$ is a binary indicator denoting whether the individual was assigned to the program (intention-to-treat). The outcome variable $Y$ represents the total individual earnings during the 30 months following program assignment, adjusted by subtracting the average cost of the program (774 dollars) for individuals with $D=1$. This adjustment accounts for resource limitations by defining the treatment effect as positive when the benefits exceed the costs, rather than simply when the benefits are positive\footnote{My framework may not accommodate other types of budget constraints; for example, it cannot handle scenarios where only a fixed amount of resources is available for program implementation.}.

The threshold policy is implemented by considering the individual's earnings in the year preceding the assignment as the index $X$. Treatment is exclusively assigned to workers with prior earnings below the threshold, based on the expectation that program services yield a more substantial positive effect for individuals who previously experienced lower earnings. Experimental data are employed to determine the threshold beyond which the treatment, on average, harms the recipients.

To estimate the SWM policy, I use the cumulative distribution function of the standard normal distribution as the kernel function, maintaining the same assumptions as in the Monte Carlo simulations. The bandwidth is chosen using the EWM policy $\hat{t}^e_n$ to compute $\hat{A}_{n}$ and $\hat{K}_{n}$ (see Appendix \ref{app:ci} for the formulas), the optimal $\hat{\lambda}_n^*$, and then the optimal bandwidth $\hat{\sigma}_n^*$. Table \ref{tab:application} reports the threshold estimates, including the confidence intervals constructed as discussed in Appendix \ref{app:ci}. The threshold with the EWM policy is almost 500 dollars lower than with the SWM (3,107 vs 3,592 dollars), a drop of $13.5\%$. The lower threshold implies that the treatment would target fewer workers: if the EWM policy were implemented, $63.55 \%$ of the workers in the sample would receive the program services, compared to the $66.19 \%$ with the SWM policy, resulting in a $3$ percentage point difference.

\begin{table}[!htbp] \centering \small
\begin{tabular}{lcc} 
\\[-1.8ex]\hline 
\hline \\[-1.8ex] 
& EWM & SWM \\
\hline \\[-1.8ex] 
 Threshold & 3107 & 3592 \\
 Confidence Interval & (2090.51, 4123.49) & (1350.34, 5983.66) \\
 Bootstrapped Confidence Interval & (1999, 4050.97) &  \\
 Asymptotic Mean Regret & 35.27 & 93.48 \\
Asymptotic Median Regret & 16.63 & 50.74 \\
\% of workers treated & 63.55 & 66.19 \\
\hline \\[-1.8ex] 
\end{tabular}
  \caption{\small Summary of the Empirical Welfare Maximizer (EWM) and Smoothed Welfare Maximizer (SWM) policies. } 
  \label{tab:application} 
\end{table}

Since I account for the costs of the program, the finding that the optimal threshold policy excludes certain workers from treatment does not imply that the program has a negative impact on those excluded. Rather, it likely reflects that, for workers with higher initial earnings, the program’s benefits are outweighed by its costs. However, if policymakers have different welfare objectives, they may still conclude that treating all individuals is optimal.

For these reasons, the numbers in the table should be considered with care, and clearly, the intention of this empirical illustration was not to advocate for a specific new job-training policy. Rather, the application aimed to assess if the EWM and SWM policies may have implications with relevant economic differences. Results suggest this is the case: together with theoretical and simulation findings, this implies that the choice between the EWM and the SWM policy should be thoughtfully considered, as it may determine relevant improvement in population welfare.

\section{Conclusion} \label{sec:conclusion}

In this paper, I addressed the problem of using experimental data to estimate optimal threshold policies when the policymaker seeks to minimize the regret associated with implementing the policy in the population. I first examined the Empirical Welfare Maximizer threshold policy, deriving its asymptotic distribution, and showing how it links to the asymptotic distribution of its regret. I then introduced the Smoothed Welfare Maximizer policy, replacing the indicator function in the EWM policy with a smooth kernel function. Under the assumptions commonly made in the policy learning literature, the convergence rate for the worst-case regret of the SWM is faster than with the EWM policy. Monte Carlo simulations corroborated the asymptotic finding that the SWM policy may perform better than the commonly studied EWM policy also in finite sample. An empirical illustration displayed that the implications of the two policies can remarkably differ in real-world application.

Three sets of problems remain open for future research, to extend the findings of this paper in diverse directions. First, while my results on the rate improvement for the EWM and SWM policies are pointwise, one might ask whether they could be formalized as optimality statements. Specifically, is it possible to establish minimax rate optimality for the two policies over some family of distributions $\mathcal{P}$? For the EWM, this would relate closely to Theorems 2.3 and 2.4 in \cite{kitagawa2018should}, as the assumptions underlying my $n^{\frac{2}{3}}$ rate improvement are connected to their margin assumption. For the SWM, however, the connection with the margin assumption is less clear, and alternative approaches may be considered.

Second, it would be interesting to extend the smoothing approach of the EWM policy to other policy classes. While threshold policies are convenient as they depend on a single parameter, the same intuition for smoothing the indicator function could also apply to linear index or multiple index policies, albeit with more complex derivations. Key questions include how the theory developed in this paper could be adapted to these policy classes and whether this approach might be generalized to all cases where the EWM policy is applicable. In some of these cases, the EWM policy is known to lack good computational properties. Since the SWM policy implies a smooth objective function, investigating its potential computational advantages would also be worthwhile.

Lastly, the framework developed in this paper for using experimental data to estimate optimal policies could inform experimental design. While the existing literature mainly focuses on optimal design for estimating the average treatment effect, it could be valuable to consider scenarios where estimating the threshold policy is the goal: how should the experimental design be adapted? How the allocation of units to treatment and control groups would change? The results presented in this paper, elucidating the connection between the distribution $P$ and the regret of the policy, provide a natural foundation for exploring experimental designs optimal for threshold policy estimation.

\appendix

\bibliographystyle{chicago}
\bibliography{references}

\section{Confidence Intervals for Threshold Policies} \label{app:ci}

Results derived in Section \ref{sec:formal} can be used to construct confidence intervals that asymptotically cover the optimal threshold policy with a given probability, and to conduct hypotheses tests. It is important to remark that, in a decision problem setting, hypotheses testing does not have a clearly motivated justification, and indeed, statistical decision theory is the alternative approach to deal with decisions under uncertainty, as pointed out in \cite{manski2021econometrics}. Rather than advocating for confidence intervals and hypothesis tests for threshold policies, this appendix aims to provide a procedure agnostic on why one may be interested in it.

For the EWM policy, \cite{rai2018statistical} proposes some confidence intervals uniformly valid for several policy classes. They rely on test inversion of a certain bootstrap procedure, which compares the welfare generated by all the policies in the class. My procedure is much simpler for the EWM threshold policies, and I directly construct confidence intervals from the asymptotic distributions derived in Theorem \ref{thm:asyd}. An analogous approach, built over results in Theorem \ref{thm:asynormsmooth}, is then used to construct confidence intervals for the SWM policy.

\subsection{Empirical Welfare Maximizer Policy}

Consider the asymptotic distribution for the EWM threshold policy derived in Theorem \ref{thm:asyd}:
\begin{gather}
    n^{1 / 3}\left(\hat{t}^e_n-t^*\right) \rightarrow^d (2\sqrt{K}/H)^{\frac{2}{3}}\argmax_r \left(B(r) - r^2 \right).
\end{gather}
If $H$ and $K$ were known, confidence intervals for the optimal policy $t^*$ with asymptotic coverage $1- \alpha$ could be constructed as $(\hat{t}^e_n - w^e_n, \hat{t}^e_n + w^e_n)$, where
\begin{gather}
    w^e_n = n^{-\frac{1}{3}} (2\sqrt{K}/H)^{\frac{2}{3}}c_{\alpha/2}
\end{gather}
and $c_{\alpha/2}$ is the critical value, the upper $\alpha/2$ quantile of the distribution of $\max_{r} B(r) - r^2$.

In practice, $H$ and $K$ are unknown and should be estimated. They are defined as:
\begin{align*}
    K
    =& f_x(t^*) \left(\frac{1}{p(t^*)} \E[Y_1^2|X = t^*] + \frac{1}{1-p(t^*)} \E[Y_0^2|X = t^*] \right) \\
    H =&  f_x(t^*) \left(\frac{\partial \E\left[Y_1 - Y_0 | X = t^* \right]}{\partial X} \right).
\end{align*}
and can be estimated by a plug-in method: consider kernel density estimator $\hat f_x(x)$ for $f_x(x)$, and local linear estimators $\hat \kappa_j(x)$ and $\hat \nu_j'(x)$ for $\kappa_j(x) = \E\left[Y_j^2 | X =x, D=j \right]$ and $\nu_j'(x) = \frac{\partial \nu_j(x)}{\partial x} = \frac{\partial \E\left[Y_j | X =x, D=j \right]}{\partial x}$. Recall that $p(x) = \E[D|X=x]$: given a value $x$, the expectation is known. Define estimators $\hat{K}_n$ and $\hat H_n$ by:
\begin{gather} \label{eq:K}
    \hat{K}_n = \hat f_x(\hat{t}^e_n) \left( \frac{1}{p(\hat{t}^e_n)} \hat \kappa_1(\hat{t}^e_n) + \frac{1}{1-p(\hat{t}^e_n)} \hat \kappa_0(\hat{t}^e_n) \right)
\end{gather}
and
\begin{gather} \label{eq:H}
    \hat H_n = \hat f_x(\hat{t}^e_n)(\hat \nu_1'(\hat{t}^e_n) - \hat \nu_0'(\hat{t}^e_n)).
\end{gather}
Under the additional assumption that the second derivatives of $f_x$, $\nu_1$ and $\nu_0$ are continuous and bounded in a neighborhood of $t^*$, and with the proper choice of bandwidth sequences, $\hat{K}_n$ and $\hat H_n$ are consistent estimators for $K$ and $H$.

Feasible confidence intervals with asymptotic coverage $1- \alpha$ can hence be constructed as $(\hat{t}^e_n - \hat{w}^e_n, \hat{t}^e_n + \hat{w}^e_n) $, where
\begin{gather}
    \hat{w}^e_n = n^{-\frac{1}{3}} (2\sqrt{\hat{K}_n}/\hat{H}_n)^{\frac{2}{3}}c_{\alpha/2}.
\end{gather}

\subsubsection{Bootstrap}
To avoid relying on tabulated values for $c_{\alpha/2}$ and on estimation of $K$, an alternative approach to inference for the EWM policy is the bootstrap. Nonparametric bootstrap is not valid for $\hat{t}^e_n$ and, more generally, for ``cube root asymptotics'' estimators \citep{abrevaya2005bootstrap,leger2006bootstrap}. Nonetheless, \cite{cattaneo2020bootstrap} provide a consistent bootstrap procedure for estimators of this type. Consistency is achieved by altering the shape of the criterion function defining the estimator whose distribution must be approximated. The standard nonparametric bootstrap is inconsistent for $Q_0(r) = -\frac{1}{2}Hr$ as defined in the proof of Theorem \ref{thm:asyd}, and hence the procedure in \cite{cattaneo2020bootstrap} directly estimates this non-random part.

Let $\{Z_i^b\}$ be a random sample from the empirical distribution $P_n$, and define the estimator $\hat{t}^b_n$ as:
\begin{gather} \label{eq:estboot}
    \hat{t}^b_n=\argmax_{t} \frac{1}{n} \sum_{i=1}^n \left[\left(\frac{D_i^b Y_i^b}{p(X_i^b)} - \frac{(1-D_i^b) Y_i^b}{(1-p(X_i^b))} \right) \ind\{X_i^b > t\}\right] - \\
        \frac{1}{n} \sum_{i=1}^n \left[\left(\frac{D_i Y_i}{p(X_i)} - \frac{(1-D_i) Y_i}{(1-p(X_i))} \right) \ind\{X_i > t\}\right] - \frac{1}{2} (t - \hat{t}_n)^2 \hat{H}_n.
\end{gather}
The bootstrap procedure proposed by \cite{cattaneo2020bootstrap} is the following:
\begin{enumerate}
    \item Compute $\hat{t}^e_n$ as described in equation \eqref{eq:estimator}.
    \item Using $\hat{t}^e_n$, compute $\hat{H}_n$ as described in equation \eqref{eq:H}.
    \item Using $\hat{t}^e_n$, $\hat{H}_n$, and the bootstrap sample $\{Z_i^b\}$, compute $\hat{t}^b_n$ as described in equation \eqref{eq:estboot}.
    \item Iterate step 3 to obtain the distribution of $n^{\frac{1}{3}} \left(\hat{t}^b_n - \hat{t}^e_n\right)$, and use it as an estimate for the distribution of $n^{\frac{1}{3}} \left(\hat{t}^e_n - t^* \right)$.
\end{enumerate}

To be valid, the procedure needs an additional assumption.
\begin{ass} \label{ass:4mom}
     {\normalfont (Bounded 4th moment)} Potential outcomes distribution are such that $\frac{1}{n^{\frac{2}{3}}}\E[Y_1^4|X=t^*] = o(1)$ and $\frac{1}{n^{\frac{2}{3}}}\E[Y_0^4|X=t^*] = o(1)$.
\end{ass}
Assumption \ref{ass:4mom} guarantees that the envelope $G_R$ is such that $PG_R^4 = o(R^{-1})$. Theorem \ref{thm:boots} proves that the distribution of $n^{\frac{1}{3}} \left(\hat{t}^b_n - \hat{t}^e_n\right)$ consistently estimates the distribution of $n^{\frac{1}{3}} \left(\hat{t}^e_n - t^* \right)$, and validate the bootstrap procedure.

\begin{thm} \label{thm:boots}
Consider estimators $\hat{t}^e_n$ defined in equation \eqref{eq:estimator} and $\hat{t}^b_n$ defined in equation \eqref{eq:estboot} and the estimand $t^*$ defined in equation \eqref{eq:estimand}. Under Assumptions \ref{ass:identification}, \ref{ass:consistency} (with $s=2$), \ref{ass:asymptotic}, and \ref{ass:4mom}, as $\hat{H}_n \rightarrow^p H$ and $n\rightarrow \infty$,
\begin{gather}
    n^{1 / 3}\left(\hat{t}^b_n - \hat{t}_n\right) \rightarrow^d (2\sqrt{K}/H)^{\frac{2}{3}}\argmax_r \left(B(r) - r^2 \right)
\end{gather}
where the limiting distribution is the same as in Theorem \ref{thm:asyd}.
\end{thm}

Distribution of $n^{\frac{1}{3}} \left(\hat{t}^b_n - \hat{t}^e_n\right)$ can hence be used to construct asymptotic valid confidence intervals and run hypothesis tests for $\hat{t}^e_n$.

\subsection{Smoothed Welfare Maximizer Policy}

Consider the asymptotic distribution for the SWM threshold policy derived in Theorem \ref{thm:asynormsmooth}, for $n \sigma_n^{2h + 1} \rightarrow \lambda < \infty$:
\begin{gather}
    (n\sigma_n)^{\frac{1}{2}}(\hat{t}^s_n - t^*) \rightarrow^d \mathcal{N}(\lambda^{\frac{1}{2}}H^{-1}A, H^{-2}\alpha_2 K).
\end{gather}
$\lambda$, $\sigma_n$, and $\alpha_2$ are known. If also $K$, $H$, and $A$ were known, confidence intervals for the optimal policy $t^*$ with asymptotic coverage $1- \alpha$ could be constructed as $(\hat{t}^s_n - b_n - w^s_n, \hat{t}^s_n - b_n + w^s_n)$, where
\begin{gather}
    b_n = (n\sigma_n)^{-\frac{1}{2}} \lambda^{\frac{1}{2}} \frac{A}{H} \\
    w^s_n = (n\sigma_n)^{-\frac{1}{2}} (\sqrt{\alpha_2 K}/H)c_{\alpha/2}
\end{gather}
and $c_{\alpha/2}$ the upper $\alpha/2$ quantile of the standard normal distribution.

In practice, $K$, $H$, and $A$ are unknown and should be estimated. As usual with inference involving bandwidths and kernels, two approaches are available: estimate and remove the asymptotic bias, or undersmooth.

For the first approach, consider estimators in equation \eqref{eq:K} for $\hat{K}_n$ and in equation \eqref{eq:H} for $\hat{H}_n$, substituting $\hat{t}^e_n$ with $\hat{t}^s_n$. For $A$, recall that
\begin{align}
    A =& -\frac{1}{h!} \alpha_1 \int_y \left(Y_1 - Y_0 \right) \varphi^{h}_x(y,t^*) d y \\
    =& -\frac{1}{h!} \alpha_1 \left[ 2 f'_x(t^*)[\lambda'_1(t^*) - \lambda'_0(t^*)] +  f_x(t^*)[\lambda''_1(t^*) - \lambda''_0(t^*)] \right]
\end{align}
where $f_x(x)$ is the probability density function of $X$ and $\nu_j(x) = \E[Y_j|X=x,D=j]$. Consider kernel density estimator $\hat f_x(x)$ and $\hat f'_x(x)$ for $f_x(x)$ and $f'_x(x)$, and local linear estimators $\hat \nu_j'(x)$ and $\hat \nu_j''(x)$ for $\nu_j'(x) = \frac{\partial \nu_j(x)}{\partial x}$ and $\nu_j''(x) = \frac{\partial^2 \nu_j(x)}{\partial^2 x}$. Define estimator $\hat{A}_n$ by:
\begin{align}
    \hat{A}_n = -\frac{1}{h!} \alpha_1 \left[ 2 \hat f'_x(\hat{t}^s_n)[\hat \lambda'_1(\hat{t}^s_n) - \hat \lambda'_0(\hat{t}^s_n)] +  \hat f_x(\hat{t}^s_n)[\hat \lambda''_1(\hat{t}^s_n) - \hat \lambda''_0(\hat{t}^s_n)] \right]
\end{align}
which consistently estimate $A$ under the additional assumption that the third derivatives of $f_x$, $\nu_1$ and $\nu_0$ are continuous and bounded in a neighborhood of $t^*$, and with the proper choice of bandwidth sequences.

Confidence intervals with asymptotic coverage $1- \alpha$ can hence be constructed as $(\hat{t}^s_n - \hat{b}_n - \hat{w}^s_n, \hat{t}^s_n - \hat{b}_n + \hat{w}^s_n) $, where
\begin{gather}
    \hat{b}_n = (n\sigma_n)^{-\frac{1}{2}} \lambda^{\frac{1}{2}} \frac{\hat{A}_n}{\hat{H}_n} \\
    \hat{w}^s_n = (n\sigma_n)^{-\frac{1}{2}} (\sqrt{\alpha_2 \hat{K}_n}/\hat{H}_n)c_{\alpha/2}
\end{gather}

The second approach relies on undersmoothing, and chooses a suboptimally small $\sigma_n$ to eliminate the asymptotic bias, with no need to estimate $A$. Instead of a bandwidth sequence $\sigma_n = O_p(n^{-\frac{1}{2h +1}})$, it considers a sequence $\sigma_n = o_p(n^{-\frac{1}{2h +1}})$ such that $n \sigma_n^{2h + 1} \rightarrow \lambda =0$, and ensures $b_n \to 0$. Confidence intervals with asymptotic coverage $1- \alpha$ can hence be constructed as $(\hat{t}^s_n - \hat{w}^s_n, \hat{t}^s_n + \hat{w}^s_n)$, with $\hat{w}^s_n$ defined as above.

\section{Monte Carlo Simulations} \label{app:mc}
I report the analogous of Table \ref{tab:median_regret} and Figures \ref{fig:median_reg_model1} and \ref{fig:median_reg_model2} for the expected regret. Comments made for the median also extend to the expected regret, with the caveat that my results do not imply that the expected value of the finite sample regret converges to the expected value of the asymptotic distribution.

\begin{table}[!htbp] \centering \small
  \caption{ \small Finite sample and asymptotic expected regrets with the EWM and SWM policies across different models are presented. Finite sample (empirical) values are computed through 5,000 Monte Carlo simulations. The last column displays the ratio between finite sample expected regrets with the EWM and SWM policies.} 
  \label{tab:exp_regret} 
\begin{tabular}{@{\extracolsep{5pt}} cccccccc} 
\\[-1.8ex]\hline 
\hline \\[-1.8ex] 
 Model & n & \multicolumn{2}{c}{EWM} & \multicolumn{3}{c}{SWM} & Ratio  \\
\hline
& & empirical & asymptotic & empirical $\sigma_n^*$ & empirical $\hat{\sigma}_n^*$ & asymptotic \\
\hline \\[-1.8ex] 
\multirow{4}{*}{1} & $500$ & $89.635$ & $96.190$ & $31.492$ & $77.747$ & $39.714$ & $1.153$ \\ 
& $1,000$ & $58.799$ & $60.596$ & $18.248$ & $44.240$ & $22.809$ & $1.329$ \\ 
& $2,000$ & $37.297$ & $38.173$ & $10.887$ & $29.176$ & $13.101$ & $1.278$ \\ 
& $3,000$ & $28.615$ & $29.131$ & $7.872$ & $16.189$ & $9.471$ & $1.768$ \\ 
\hline
\multirow{4}{*}{2} & $500$ & $276.867$ & $400.166$ & $215.876$ & $291.654$ & $439.442$ & $0.949$ \\ 
& $1,000$ & $201.182$ & $252.089$ & $151.181$ & $209.069$ & $252.393$ & $0.962$ \\ 
& $2,000$ & $149.572$ & $158.806$ & $107.733$ & $151.069$ & $144.962$ & $0.990$ \\ 
& $3,000$ & $124.168$ & $121.192$ & $90.251$ & $125.565$ & $104.805$ & $0.989$ \\
\hline \\[-1.8ex] 
\end{tabular}
\end{table}

\begin{figure}
\begin{center}
\resizebox{0.8\textwidth}{!}{
\includegraphics{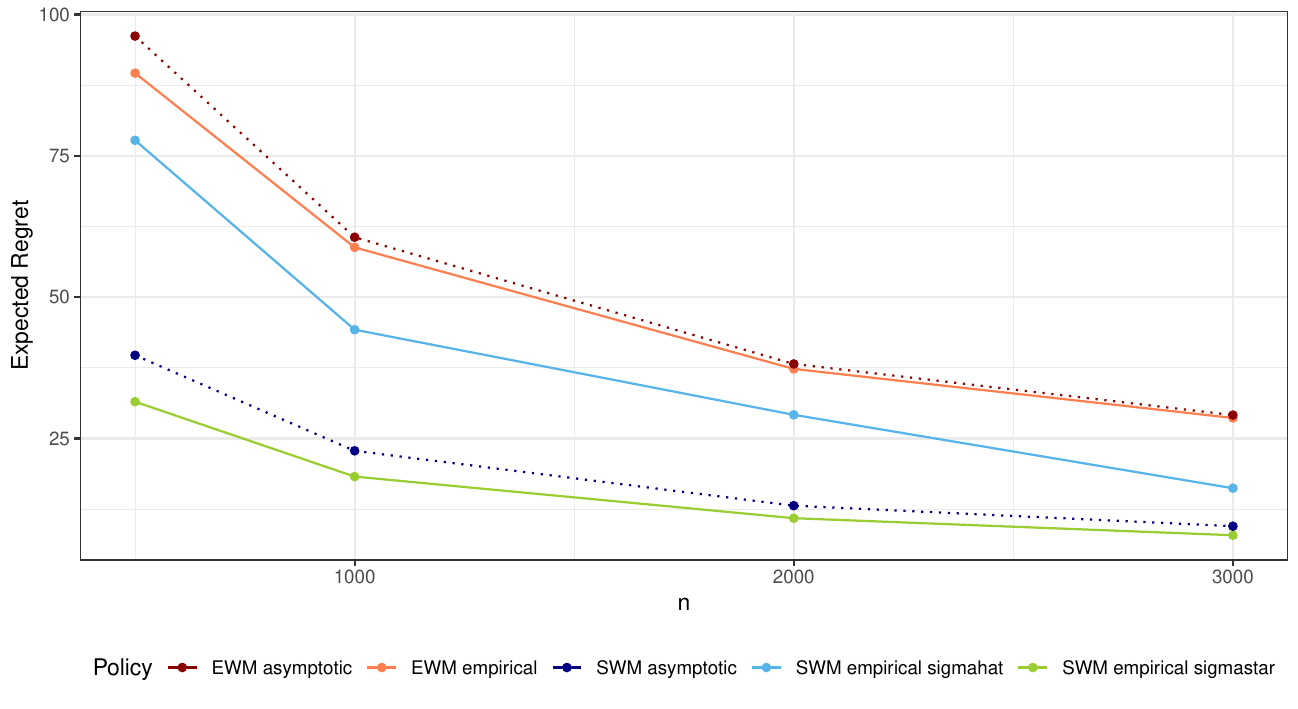}
}
\caption{\small The figure illustrates asymptotic and finite sample expected regrets for the EWM and SWM policies in Model 1, corresponding to the values reported in Table \ref{tab:exp_regret}.}
\label{fig:exp_reg_model1}
\end{center}
\end{figure}

\begin{figure}
\begin{center}
\resizebox{0.8\textwidth}{!}{
\includegraphics{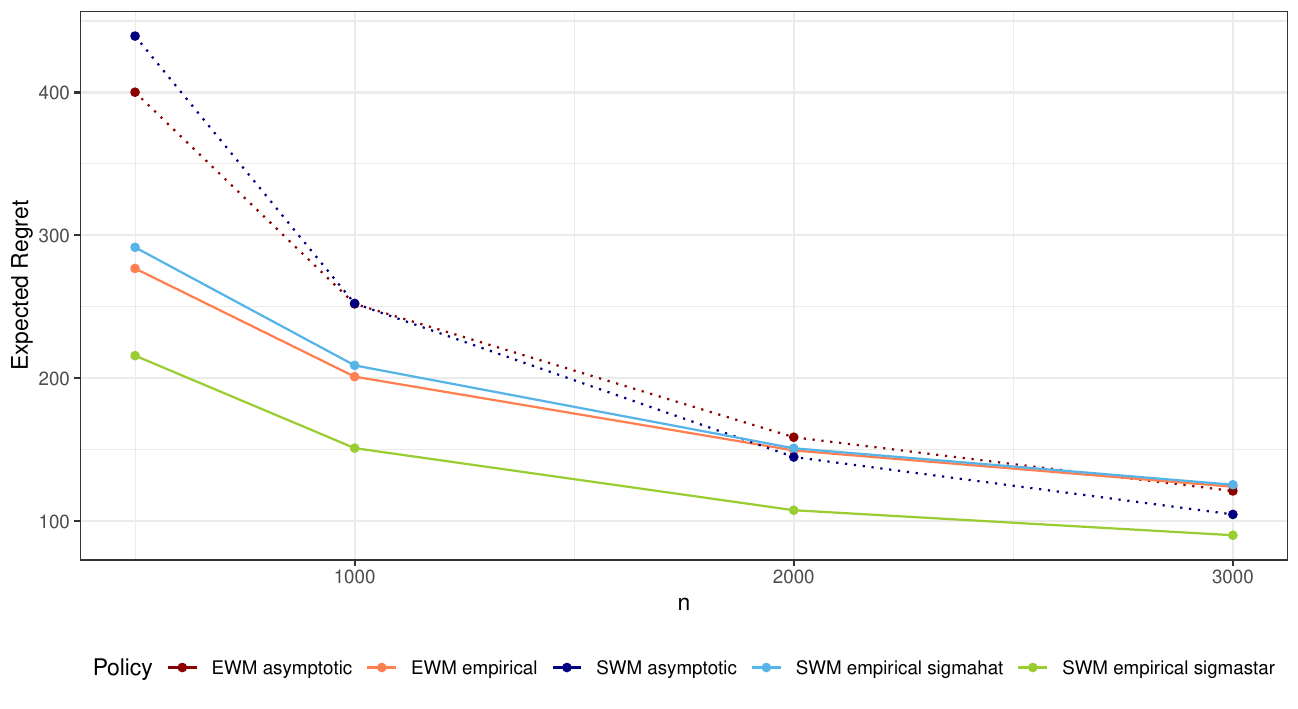}
}
\caption{\small The figure illustrates asymptotic and finite sample expected regrets for the EWM and SWM policies in Model 2, corresponding to the values reported in Table \ref{tab:exp_regret}.}
\label{fig:exp_reg_model2}
\end{center}
\end{figure}

\newpage
\section{Proofs}
\setcounter{thm}{0}
\subsection*{Theorem \ref{thm:con}}
\begin{thm}
Consider the EWM policy $\hat{t}^e_n$ defined in equation \eqref{eq:estimator} and the optimal policy $t^*$ defined in equation \eqref{eq:estimand}. Under Assumptions \ref{ass:identification} and \ref{ass:consistency} (with $s=0$),
\begin{gather*}
    \hat{t}^e_n \rightarrow^{a.s.} t^*
\end{gather*}
i.e. $\hat{t}^e_n$ is a strongly consistent estimator for $t^*$.
\end{thm}
\begin{proof}
Estimator $\hat{t}^e_n$ in \eqref{eq:estimator} can be written as
\begin{align}
    \hat{t}^e_n= & \argmax_{t} \frac{1}{n} \sum_{i=1}^n \left[\left(\frac{D_i Y_i}{p(X_i)} - \frac{(1-D_i) Y_i}{(1-p(X_i))} \right) (\ind\{X_i > t\} - \ind\{X_i > t^* \}) \right]  \\
    =& \argmax_{t} \frac{1}{n} \sum_{i=1}^n m\left(Z_i, t\right) 
\end{align}
where
\begin{gather}
    m\left(Z, t\right)=\left(\frac{D Y}{p(X)} - \frac{(1-D) Y}{(1-p(X))} \right) \left(\ind\{X > t\} - \ind\{X > t^*\} \right)
\end{gather}
and $\{Z_i\}$ is a sample of $n$ observation from distribution $P$.

Define $P_n m\left(\cdot, t\right) = \sum_{i=1}^n m\left(Z_i, t\right)$ and $P m\left(Z, t\right) = \E_P[ m\left(Z, t\right)]$. With this notation, equation \eqref{eq:estimand} and \eqref{eq:estimator} can be written as
\begin{gather}
    t^* = \argmax_{t} P m\left(Z, t\right) \\
    \hat{t}^e_n= \argmax_{t} P_n m\left(\cdot, t\right).
\end{gather}

Threshold policies I am considering can be seen as a tree partition of depth 1. Tree partitions of finite depth are a VC class \citep{leboeuf2020decision}, and hence $m\left(\cdot, t\right)$ is a manageable class of functions. Consider the envelope function $F = 2\left|\frac{D Y}{p(X)} - \frac{(1-D) Y}{(1-p(X))}\right|$. Note that $\E\left[\left|\frac{D Y}{p(X)} - \frac{(1-D) Y}{(1-p(X))}\right|^2\right] = \frac{1}{p(X)} \E[Y_1^2] + \frac{1}{1-p(X)} \E[Y_0^2] $: Assumption \ref{ass:sqintegr} guarantees the existence of $\E\left[\left|\frac{D Y}{p(X)} - \frac{(1-D) Y}{(1-p(X))}\right|^2\right]$.

It follows from corollary 3.2 in \cite{kim1990cube} that
\begin{gather}
    \sup _t\left|P_n m(\cdot, t)-P m(Z, t)\right| \rightarrow^{a.s.} 0 .
\end{gather}
Under Assumptions \ref{ass:uniq} and \ref{ass:smoot}, $P m(Z, t)$ is continuous in $t$, and $t^*$ is the unique maximizer. Hence,
\begin{align}
    & \sup _t\left|P_n m(\cdot, t)-P m(Z, t)\right| + P m(Z, t^*) \geq \\
   & \left|P_n m(\cdot, \hat{t}^e_n)-P m(Z, \hat{t}^e_n)\right| + P m(Z, t^*) \geq \\
   & \left|P_n m(\cdot, \hat{t}^e_n)-P m(Z, \hat{t}^e_n)\right| + P m(Z, \hat{t}^e_n) \geq \\
   & P_n m(\cdot, \hat{t}^e_n) \geq P_n m(\cdot, t^*) \rightarrow P m(Z, t^*)
\end{align}
where the second inequality is due to the fact that $t^*$ is the maximizer of $Pm(Z,t)$, the third comes from the triangular inequality, the fourth from $\hat{t}^e_n$ being the maximizer of $P_n m(\cdot, t)$, and the last limit comes from LLN. This prove that $ P_n m(\cdot, \hat{t}^e_n) \rightarrow^{a.s.} P m(Z, t^*) $, and hence $ P m(Z, \hat{t}^e_n) \rightarrow^{a.s.} P m(Z, t^*) $. Since $t^*$ is the unique maximizer of $P m(Z, t)$ and $P m(Z, t)$ is continuous, $\hat{t}^e_n \rightarrow^{a.s.} t^*$. It means that $\hat{t}^e_n$ is a strongly consistent estimator for $t^*$.
\end{proof}

\subsection*{Theorem \ref{thm:asyd}}
\begin{thm}
Consider the EWM policy $\hat{t}^e_n$ defined in equation \eqref{eq:estimator} and the optimal policy $t^*$ defined in equation \eqref{eq:estimand}. Under Assumptions \ref{ass:identification}, \ref{ass:consistency} (with $s=2$), and \ref{ass:asymptotic}, as $n\rightarrow \infty$,
\begin{gather}
    n^{1 / 3}\left(\hat{t}^e_n-t^*\right) \rightarrow^d (2\sqrt{K}/H)^{\frac{2}{3}}\argmax_r \left(B(r) - r^2 \right)
\end{gather}
where $B(r)$ is the two-sided Brownian motion process, and $K$ and $H$ are
\begin{align*}
    K
    =& f_x(t^*) \left(\frac{1}{p(t^*)} \E[Y_1^2|X = t^*] + \frac{1}{1-p(t^*)} \E[Y_0^2|X = t^*] \right) \\
    H =&  f_x(t^*) \left(\frac{\partial \E_P\left[Y_1 - Y_0 | X = t^* \right]}{\partial X} \right).
\end{align*}
\end{thm}
\begin{proof}
The proof shows that conditions for the main theorem in \cite{kim1990cube} hold; hence, their result is valid for $\hat{t}^e_n$. For completeness, I report the theorem.

\begin{thmkp} \label{thm:kp}
Consider estimators defined by maximization of processes
\begin{gather}
    P_n g(\cdot, \theta)=\frac{1}{n} \sum_{i=1}^n g\left(\xi_i, \theta\right)
\end{gather}
where $\left\{\xi_i\right\}_i$ is a sequence of i.i.d. observations from a distribution $P$ and $\{g(\cdot, \theta): \theta \in \Theta\}$ is a class of functions indexed by a subset $\Theta$ in $\mathbb{R}^k$. The envelope $G_R(\cdot)$ is defined as the supremum of $g(\cdot, \theta)$ over the class
\begin{gather}
    \mathcal{G}_R=\left\{|g(\cdot, \theta)|:\left|\theta-\theta_0\right| \leq R\right\}, \quad R>0.
\end{gather}

Let $\left\{\theta_n\right\}$ be a sequence of estimators for which
\begin{enumerate}
    \item $P_n g\left(\cdot, \theta_n\right) \geq \sup _{\theta \in \Theta} P_n g(\cdot, \theta)-o_P\left(n^{-2 / 3}\right)$.
    \item $\theta_n$ converges in probability to the unique $\theta_0$ that maximizes $P g(\cdot, \theta)$.
    \item $\theta_0$ is an interior point of $\Theta$.
\end{enumerate}
Let the functions be standardized so that $g\left(\cdot, \theta_0\right) \equiv 0$. If the classes $\mathcal{G}_R$ for $R$ near 0 are uniformly manageable for the envelopes $G_R$ and satisfy:
\begin{enumerate} \setcounter{enumi}{3}
    \item $P g(\cdot, \theta)$ is twice differentiable with second derivative matrix $-H$ at $\theta_0$.
    \item $K(s, r)=\lim _{\alpha \rightarrow \infty} \alpha P g\left(\cdot, \theta_0+r / \alpha\right) g\left(\cdot, \theta_0+s / \alpha\right)$ exists for each $s, r$ in $\R^k$ and \\ $\lim _{\alpha \rightarrow \infty} \alpha P g\left(\cdot, \theta_0+r / \alpha\right)^2\left\{\left|g\left(\cdot, \theta_0+r / \alpha\right)\right|>\varepsilon \alpha\right\}=0$ for each $\varepsilon>0$ and $r$ in $\mathbb{R}^k$.
    \item $PG_R^2 = O(R)$ as $R\rightarrow 0$ and for each $\varepsilon>0$ there is a constant C such that $PG_R^2 \ind\{G_R > C\} \leq \varepsilon R$ for $R$ near 0.
    \item $P\left|g\left(\cdot, \theta_1\right)-g\left(\cdot, \theta_2\right)\right|=O\left(\left|\theta_1-\theta_2\right|\right)$ near $\theta_0$.
\end{enumerate}
Then, the process $n^{2 / 3} P_n g\left(\cdot, \theta_0+r n^{-1 / 3}\right)$ converges in distribution to a Gaussian process $Q(r)$ with continuous sample paths, expected value $-\frac{1}{2} r^{\prime} H r$ and covariance kernel $K$. If $H$ is positive definite and if $Q$ has nondegenerate increments, then $n^{1 / 3}\left(\theta_n-\theta_0\right)$ converges in distribution to the (almost surely unique) random vector that maximizes $Q$.
\end{thmkp}

I apply Theorem 1.1 in \cite{kim1990cube} by taking $\xi_i=Z_i, \theta=t, \theta_n=\hat{t}^e_n, \theta_0=t^*$, $g(\cdot, \theta) = m(\cdot,t)$, where $m(\cdot,t)$ is standardized:
\begin{gather}
    m\left(Z_i, t\right)=\left(\frac{D_i Y_i}{p(X_i)} - \frac{(1-D_i) Y_i}{(1-p(X_i))} \right) \left(\ind\{X_i > t\} - \ind\{X_i > t^*\} \right).
\end{gather}

First, I will verify that conditions 1-7 apply to my setting:
\begin{enumerate}
    \item $P_n m(\cdot, t)$ takes only finite ($n+1$) values; hence, condition 1 is satisfied with the equality.
    \item In Theorem \ref{thm:con}, I proved that $\hat{t}^e_n$ is a strongly consistent estimator for $t^*$.
    \item $t^*$ is an interior point of $\mathcal{T}$ by Assumption \ref{ass:uniq}.
\end{enumerate}

I need to prove that the classes $\mathcal{G}_R$ for $R$ near 0 are uniformly manageable for the envelopes $G_R$. The envelope function $G_R(\cdot)$ is defined as
\begin{align}
    G_R(z) & =\sup \left\{m(z,t):\left|t-t^*\right| \leq R\right\} \\
    & =\sup _{\left|t-t^*\right| \leq R}\left[\left(\frac{d y}{p(X)} - \frac{(1-d) y}{(1-p(X))} \right) \left(\ind\{x > t\} - \ind\{x > t^*\} \right)\right] \\
    & =\left|\frac{d y}{p(X)} - \frac{(1-d) y}{(1-p(X))}\right| \ind\{\left|x-t^*\right| < R\}
\end{align}
and I have:
\begin{align} \label{eq:pgrO}
    P G_R^2 & =\E \left[ \left(\frac{DY}{p(X)} - \frac{(1-D) Y}{(1-p(X))}\right)^2 \ind\{\left|X-t^*\right| < R\} \right]\\
    & = \E \left[ \left(\frac{DY}{p(X)} - \frac{(1-D) Y}{(1-p(X))}\right)^2 \Bigg| X \in (t^*-R,t^*+R) \right]\\
    & = 2 R\E \left[ \left(\frac{DY}{p(X)} - \frac{(1-D) Y}{(1-p(X))}\right)^2 \Bigg| X =t^* \right] + o(1) = O(R)
\end{align}
where the second to last equality comes from Assumption \ref{ass:sqintegr}. The envelope function is uniformly square-integrable for $R$ near 0, and therefore, the classes $\mathcal{G}_R$ are uniformly manageable.

\begin{enumerate} \setcounter{enumi}{3}
    \item Define $h(t) = P m(Z, t)$ and consider derivatives:
    \begin{align}
         h(t) =& \E_P\left[\left(\frac{D Y}{p(X)} - \frac{(1-D) Y}{(1-p(X))} \right) (\ind\{X > t\} - \ind\{X > t^*\}) \right] = \\
        & \E_P\left[\left(Y_1 - Y_0\right) (\ind\{X > t\} - \ind\{X > t^*\})\right] \\
        h'(t) =& - f_x(t)\E_P\left[Y_1 - Y_0 \Big| X = t \right] \\
        h''(t) =& - f'_x(t)\E_P\left[Y_1 - Y_0 \Big| X = t \right] - f_x(t) \left(\frac{\partial \E_P\left[Y_1 - Y_0 | X = t \right]}{\partial X} \right). 
    \end{align}
    Assumption \ref{ass:smoot} with $s=2$ guarantees the existence of $h'$ and $h''$. Since $\E_P\left[Y_1 - Y_0 \Big| X = t^* \right] = 0$, $H$ is given by
    \begin{gather}
        H= - h''(t^*) =  f_x(t^*) \left(\frac{\partial \E_P\left[Y_1 - Y_0 | X = t^* \right]}{\partial X} \right).
    \end{gather}
    \item This condition is divided into two parts. First, I prove the existence of $K(s, r)=\lim _{\alpha \rightarrow \infty} \alpha P m\left(\cdot, t^*+r / \alpha\right) m\left(\cdot, t^*+s / \alpha\right)$ for each $s,r$ in $\R$. Covariance $K$ is:
    \begin{align}
         P m & \left(\cdot, t^*+r / \alpha\right) m\left(\cdot, t^*+s / \alpha\right) =
         \E \left[ \left(\frac{DY}{p(X)} - \frac{(1-D) Y}{(1-p(X))}\right)^2 \right. \\
        & (\ind\{X>t^*+r/ \alpha \} -\ind\{X>t^* \}) (\ind\{X>t^*+s/ \alpha \} -\ind\{X>t^* \})\Big].
    \end{align}
    If $rs<0$, covariance and $K(s,r)$ are 0. If $rs>0$, and suppose $r>0$:
    \begin{align}
         & P m\left(\cdot, t^*+r / \alpha\right) m\left(\cdot, t^*+s / \alpha\right) = \\
         & \E \left[ \left(\frac{DY}{p(X)} - \frac{(1-D) Y}{(1-p(X))}\right)^2 \Big| X \in (t^*,t^* + \min\{r,s\}/\alpha )\right]
    \end{align}
    and hence
    \begin{align}
        K(s, r)&=\lim _{\alpha \rightarrow \infty} \alpha \E \left[ \left(\frac{DY}{p(X)} - \frac{(1-D) Y}{(1-p(X))}\right)^2 \Big| X \in (t^*,t^* + \min\{r,s\}/\alpha )\right] \\
        &= \min\{r,s\} f_x(t^*) \E \left[ \left(\frac{DY}{p(X)} - \frac{(1-D) Y}{(1-p(X))}\right)^2 \Big| X = t^*\right].
    \end{align}
    where the equality is due to continuity of $f_x$ (Assumption \ref{ass:smoot} with $s=2$). Boundedness of the quantity follows from Assumptions \ref{ass:sqintegr} and \ref{ass:smoot}.
    
    Now, I will prove that $\lim _{\alpha \rightarrow \infty} \alpha P m\left(\cdot, t^* +r / \alpha\right)^2\ind \left\{\left|m\left(\cdot, t^*+r / \alpha\right)\right|>\varepsilon \alpha\right\}=0$ for each $\varepsilon>0$ and $r$ in $\R$. I have:
    \begin{align}
        & \alpha P m\left(\cdot, t^* +r / \alpha\right)^2 \ind\left\{\left|m\left(\cdot, t^*+r / \alpha\right)\right|>\varepsilon \alpha\right\} = \\
        & \alpha \E \left[ \left(\frac{DY}{p(X)} - \frac{(1-D) Y}{(1-p(X))}\right)^2 \right. \\
        &(\ind\{X>t^*+r/ \alpha \} -\ind\{X>t^* \})^2 \ind\left\{\left|m\left(\cdot, t^*+r / \alpha\right)\right|>\varepsilon \alpha\right\} \Big] = \\
        &\alpha \E \left[ \left(\frac{DY}{p(X)} - \frac{(1-D) Y}{(1-p(X))}\right)^2 \right. \\
        &(\ind\{X>t^*+r/ \alpha \} -\ind\{X>t^* \})^2
         \ind\left\{\left|\frac{DY}{p(X)} - \frac{(1-D) Y}{(1-p(X))} \right|>\varepsilon \alpha\right\} \Big]\leq \\
        &\alpha \E \left[ \left(\frac{DY}{p(X)} - \frac{(1-D) Y}{(1-p(X))}\right)^2 
         \ind\left\{\left|\frac{DY}{p(X)} - \frac{(1-D) Y}{(1-p(X))} \right|>\varepsilon \alpha\right\} \right].
    \end{align}
    Some algebra gives
    \begin{align}
        & \E \left[ \left(\frac{DY}{p(X)} - \frac{(1-D) Y}{(1-p(X))}\right)^2 \ind\left\{\left|\frac{DY}{p(X)} - \frac{(1-D) Y}{(1-p(X))} \right|> \epsilon \right\} \right] = \\
        & \E \left[ \left(\frac{DY_1^2}{p(X)^2} + \frac{(1-D) Y_0^2}{(1-p(X))^2}\right) \ind\left\{\left|\frac{DY}{p(X)} - \frac{(1-D) Y}{(1-p(X))} \right|> \epsilon \right\} \right] = \\
        & \E \left[ \frac{Y_1^2}{p(X)} \ind\left\{\left|\frac{Y_1}{p(X)} \right|> \epsilon \right\} \right] + \E \left[  \frac{ Y_0^2}{(1-p(X))} \ind\left\{\left| \frac{ Y_0}{(1-p(X))} \right|> \epsilon \right\} \right] =\\
        & \frac{1}{p(X)}\E \left[ Y_1^2 \ind\left\{\left|Y_1 \right|> \epsilon_1 \right\} \right] + \frac{ 1}{(1-p(X))}\E \left[  Y_0^2 \ind\left\{\left|  Y_0 \right|> \epsilon_0 \right\} \right]
    \end{align}
    and hence the condition is satisfied if 
    \begin{gather}
        \lim_{\alpha \rightarrow \infty} \alpha \E \left[ Y_1^2 \ind\left\{\left|Y_1 \right|> \epsilon \alpha \right\} \right] = 0 \\
        \lim_{\alpha \rightarrow \infty} \alpha \E \left[ Y_0^2 \ind\left\{\left|Y_0 \right|> \epsilon \alpha \right\} \right]= 0.
    \end{gather}
    Consider the limit for $Y_1$:
    \begin{gather}
        \lim_{\alpha \rightarrow \infty} \alpha \E \left[ Y_1^2 \ind\left\{\left|Y_1 \right|> \epsilon \alpha \right\} \right] = \lim_{\alpha \rightarrow \infty} \frac{\int_{\epsilon \alpha} Y_1^2 \varphi_1(y_1) dy_1}{\alpha^{-1}} = \lim_{\alpha \rightarrow \infty} \frac{\epsilon^3 \alpha^2 \varphi_1(\epsilon \alpha) }{\alpha^{-2}} = \\
        \lim_{\alpha \rightarrow \infty} \epsilon^3 \alpha^4 \varphi_1(\epsilon \alpha) = \lim_{\alpha \rightarrow \infty} \epsilon^3 \alpha^4 |\epsilon \alpha |^{-(4+\delta)} o(1) = 0
    \end{gather}
    where the second to last equality follows from Assumption \ref{ass:asymptotic}.
    \item I showed that $PG_R^2 = O(R)$ as $R\rightarrow 0$ in equation \eqref{eq:pgrO}. I need to prove that for each $\varepsilon>0$ there is a constant C such that $PG_R^2 \ind\{G_R > C\} \leq \varepsilon R$ for $R$ near 0. For any $\varepsilon >0$ and $C >0$:
    \begin{align}
        PG_R^2 & \ind\{G_R > C\} = \E \left[ \left(\frac{DY}{p(X)} - \frac{(1-D) Y}{(1-p(X))}\right)^2 \ind\{\left|X-t^*\right| < R\} \ind\{G_R > C\} \right] \leq \\
        &\E \left[ \left(\frac{DY}{p(X)} - \frac{(1-D) Y}{(1-p(X))}\right)^2 \ind\{\left|X-t^*\right| < R\} \ind\left\{\left|\frac{DY}{p(X)} - \frac{(1-D) Y}{(1-p(X))} \right|> C \right\} \right] \leq \\
        &\E \left[ \left(\frac{DY}{p(X)} - \frac{(1-D) Y}{(1-p(X))}\right)^2 \ind\left\{\left|\frac{DY}{p(X)} - \frac{(1-D) Y}{(1-p(X))} \right|> C \right\} \right] \rightarrow 0
    \end{align}
    where the last limit is taken for $C\rightarrow \infty$, and follows from Assumption \ref{ass:asymptotic}.
    \item I need to show that $P\left|m\left(\cdot, t_1\right)-m\left(\cdot, t_2\right)\right|=O\left(\left|t_1-t_2\right|\right)$ near $t^*$. Consider $t_2 > t_1$:
    \begin{gather}
        P\left|m\left(\cdot, t_1\right)-m\left(\cdot, t_2\right)\right| = \E\left[ \left| \left(\frac{D Y}{p(X)} - \frac{(1-D) Y}{(1-p(X))} \right) (\ind\{X > t_1\} - \ind\{X > t_2\}) \right| \right] \leq \\
        M_x \E\left[ \left| \left(\frac{D Y}{p(X)} - \frac{(1-D) Y}{(1-p(X))} \right)\right| \Big| X \in (t_1, t_2)  \right] \leq M_x M_y |t_2 - t_1|
    \end{gather}
    where $M_x = \max_{x \in (t_1, t_2)} f_x(x)$ and $M_y = \max_{x \in (t_1, t_2)} \E\left[ \left| \left(\frac{D Y}{p(X)} - \frac{(1-D) Y}{(1-p(X))} \right)\right| \Big| X  \right]$. $M_x <\infty$ and $M_y <\infty$ because of Assumption \ref{ass:smoot}.
\end{enumerate}

Assumptions 1-7 of Theorem 1.1 in \cite{kim1990cube} are hence satisfied. It follows that, for $n\rightarrow \infty$,
\begin{gather}
    n^{1 / 3}\left(\hat{t}^e_n-t^*\right) \rightarrow^d \argmax_{r}Q(r)
\end{gather}
where $Q(r) = Q_1(r) + Q_0(r)$, and $Q_1$ is a non degenerate zero-mean Gaussian process with covariance $K$, while $Q_0(r)$ is non-random and $Q_0(r) = -\frac{1}{2}r^2H$. By Assumptions \ref{ass:smoot} and \ref{ass:nonflattau}, $H\neq 0$.

Limiting distribution $\argmax_{r}Q(r)$ is of \cite{chernoff1964estimation} type. It can be shown \citep{banerjee2001likelihood} that
\begin{gather}
    \argmax_r Q(r) =^d (2\sqrt{K}/H)^{\frac{2}{3}}\argmax_r B(r) - r^2
\end{gather}
where $B(r)$ is the two-sided Brownian motion process, $K$ is:
\begin{align}
    K =& f_x(t^*) \E \left[ \left(\frac{DY}{p(X)} - \frac{(1-D) Y}{(1-p(X))}\right)^2 \Big| X = t^*\right] \\
    =& f_x(t^*) \left(\frac{1}{p(t^*)} \E[Y_1^2|X = t^*] + \frac{1}{1-p(t^*)} \E[Y_0^2|X = t^*] \right)
\end{align}
and $H$ is:
\begin{gather}
    H=  f_x(t^*) \left(\frac{\partial \E_P\left[Y_1 - Y_0 | X = t^* \right]}{\partial X} \right).
\end{gather}
This completes the proof of the theorem.
\end{proof}

\subsection*{Corollary \ref{cor:ewm_regret}}

\begin{corollary} \label{cor:ewm_regret}
    The asymptotic distribution of regret $\mathcal{R}(\hat{t}^e_n)$ is:
    \begin{gather*}
        n^{\frac{2}{3}} \mathcal{R}(\hat{t}^e_n) \rightarrow^d  \left( \frac{2K^{2}}{H} \right)^\frac{1}{3}  \left( \argmax_r B(r) - r^2 \right)^2.
    \end{gather*}
    The expected value of the asymptotic distribution is $K^\frac{2}{3} H^{-\frac{1}{3}} C^e$, where $$C^e= \sqrt[3]{2} \E\left[\left( \argmax_r B(r) - r^2 \right)^2\right]$$ is a constant not dependent on $P$.
\end{corollary}

\begin{proof}
    Result in equation \eqref{eq:asympt_regret_d} for $\hat{t}^e_n$ implies
    \begin{gather*}
    n^{\frac{2}{3}} \mathcal{R}(\hat{t}^e_n) =  \frac{1}{2} W''(\tilde{t}) \left( n^{\frac{1}{3}} \left(\hat{t}^e_n-t^*\right) \right)^2,
    \end{gather*}
    where $|\tilde{t}-t^*| \leq |\hat{t}_n - t^*|$. By continuous mapping theorem
    \begin{gather*}
        W''(\tilde{t}) \rightarrow^p W''(t^*) = H
    \end{gather*}
    and hence by Slutsky's theorem
    \begin{gather*}
    n^{\frac{2}{3}} \left( W(\hat{t}_n) - W(t^*) \right) \rightarrow^d  \left( \frac{2K^{2}}{H} \right)^\frac{1}{3}  \left( \argmax_r B(r) - r^2 \right)^2.
    \end{gather*}
\end{proof}

\subsection*{Theorem \ref{thm:cons_sm}}

\begin{thm} \label{thm:cons_sm}
Consider the SWM policy $\hat{t}^s_n$ defined in equation \eqref{eq:estimator_sm} and the optimal policy $t^*$ defined in equation \eqref{eq:estimand}. Under Assumptions \ref{ass:identification}, \ref{ass:consistency} (with $s=0$), and \ref{ass:kernel},
\begin{gather*}
    \hat{t}^s_n \rightarrow^{a.s.} t^*
\end{gather*}
i.e. $\hat{t}^s_n$ is a strongly consistent estimator for $t^*$.
\end{thm}

\begin{proof}
To prove the result, I show that conditions for Theorem 4.1.1 in \cite{amemiya1985advanced} hold, and hence $\hat{t}^s_n$ is consistent for $t^*$.

First, define function $m^s(Z,t)$:
\begin{gather*}
    m^s \left(Z, t\right)=\left(\frac{D Y}{p(X)} - \frac{(1-D) Y}{(1-p(X))} \right) \left(k\left(\frac{X - t}{\sigma_n} \right) - k\left(\frac{X - t^*}{\sigma_n} \right) \right)
\end{gather*}
and recall definitions of $m(Z,t)$, $P_n$, and $P$ introduced in the proof of Theorem \ref{thm:asyd}:
\begin{align}
    m\left(Z, t\right) &=\left(\frac{D Y}{p(X)} - \frac{(1-D) Y}{(1-p(X))} \right) \left(\ind\{X > t\} - \ind\{X > t^*\} \right) \\
    P_n m\left(\cdot, t\right) &= \sum_{i=1}^n m\left(Z_i, t\right) \\
    P m\left(Z, t\right) &= \E_P[ m\left(Z, t\right)].
\end{align}
In this notation, $\hat{t}^s_n = \argmax_{t} P_n m^s (.,t)$ and $t^* = \argmax_{t} P m (Z,t)$. I can now show that conditions $A$, $B$, and $C$ for Theorem 4.1.1 in \cite{amemiya1985advanced} hold:

\begin{itemize}
\item[A)] Parameter space $\mathcal{T}$ is compact by Assumption \ref{ass:uniq}.
\item[B)] Function $P_n m^s \left(Z_i, t\right)$ is continuous in $t$ for all $Z$ and is a measurable function of $Z$ for all $t \in \mathcal{T}$, as $k(\cdot)$ is continuous by Assumption \ref{ass:kernel}.
\item[C1)] I need to prove that $P_n m^s \left(Z_i, t\right)$ converges a.s. to $P m(Z,t)$ uniformly in $t \in \mathcal{T}$ as $n \rightarrow \infty$, i.e. $\sup _t\left|P_n m^s(\cdot, t)-P m(Z, t)\right| \rightarrow^{a.s.} 0$. Note that:
\begin{align}
    & \sup _t\left|P_n m^s(\cdot, t)-P m(Z, t)\right| \leq  \\
   & \sup _t\left|P_n m^s(\cdot, t)-P m^s(Z, t)\right| + \sup _t\left|P m^s(\cdot, t)-P m(Z, t)\right|.
\end{align}
I need to show that the two addends on the right-hand side converge to zero.

To show uniform convergence of $P_n m^s(\cdot, t)$ to $P m^s(Z, t)$, I consider sufficient conditions provided by Theorem 4.2.1 in \cite{amemiya1985advanced}. $m^s(Z,t)$ is continuous in $t\in \mathcal{T}$ with $\mathcal{T}$ compact, and measurable in $Z$. I only need to show that $\E[\sup_{t \in \mathcal{T}}|m^s(Z,t)|]<\infty$. By Assumption \ref{ass:kernel}, $k(\cdot)$ is a bounded function, i.e. it exists an $M$ such that $|k(x)|<M$ for all $x$. Hence $\E[\sup_{t \in \mathcal{T}}|m^s(Z,t)|]\leq M \E\left[\left|\frac{D Y}{p(X)} - \frac{(1-D) Y}{(1-p(X))} \right|\right]$, and $\E\left[\left|\frac{D Y}{p(X)} - \frac{(1-D) Y}{(1-p(X))} \right|\right]< \infty$ by Assumption \ref{ass:sqintegr}.

To show uniform convergence of $P m^s(\cdot, t)$ to $P m(Z, t)$, note that $t\in \mathcal{T}$, where $\mathcal{T}$ is bounded, and hence the result holds for $\sigma_n \rightarrow 0$.
\item[C2)] By Assumption \ref{ass:uniq}, $t^*$ is the unique global maximum of $P m(Z,t)$.
\end{itemize}

Assumptions $A$, $B$, $C$ of Theorem 4.1.1 in \cite{amemiya1985advanced} are satisfied, and hence $\hat{t}^s_n \rightarrow^{a.s.} t^*$.
\end{proof}

\subsection*{Lemmas}
\setcounter{thm}{0}
Proof of Theorem \ref{thm:asynormsmooth} requires some intermediate lemmas, stated and proved below. Arguments follow the ideas in \cite{horowitz1992smoothed}, but are adapted to my context. I report the entire proof for completeness, even when it overlaps with the original in \cite{horowitz1992smoothed}.

To make the notation simpler, define:
\begin{gather*}
    \hat{S}_n(t,\sigma_n) =  \frac{1}{n} \sum_{i=1}^n \left[\left(\frac{D_i Y_i}{p(X_i)} - \frac{(1-D_i) Y_i}{(1-p(X_i))} \right) k\left(\frac{X_i - t}{\sigma_n} \right) \right]
\end{gather*}
and note that $\hat{t}^s_n=\argmax_{t} \hat{S}_n(t,\sigma_n)$. Then define:
\begin{gather*}
    \hat{S}_n^1(t,\sigma_n) = \frac{\partial \hat{S}_n(t,\sigma_n)}{ \partial t} = - \frac{1}{\sigma_n} \frac{1}{n} \sum_{i=1}^n \left[\left(\frac{D_i Y_i}{p(X_i)} - \frac{(1-D_i) Y_i}{(1-p(X_i))} \right) k'\left(\frac{X_i - t}{\sigma_n} \right) \right]\\
    \hat{S}_n^2(t,\sigma_n) = \frac{\partial \hat{S}_n(t,\sigma_n)}{ \partial^2 t} = \frac{1}{\sigma_n^2} \frac{1}{n} \sum_{i=1}^n \left[\left(\frac{D_i Y_i}{p(X_i)} - \frac{(1-D_i) Y_i}{(1-p(X_i))} \right) k''\left(\frac{X_i - t}{\sigma_n} \right) \right].
\end{gather*}
Indicate with $\varphi_{y,x}(y,x)$ the joint distribution of $Y_1$, $Y_0$, and $X$, and with $\varphi_{y|x}(y|x)$ the conditional distribution, where $\varphi_{y,x}(y,x) = \varphi_{y|x}(y|x) f_x(x)$.

\subsubsection*{Lemma 1}
    \begin{lem}\label{lemma1}
        Let Assumptions \ref{ass:identification}, \ref{ass:consistency} (with $s=h + 1$ for some $h\geq 2$), and \ref{ass:asymptotic_sm} hold. Then
        \begin{gather}
            \lim_{n\rightarrow \infty} \E \left[ \sigma_n^{-h} \hat{S}_n^1(t^*,\sigma_n) \right] = A  \\
            \lim_{n\rightarrow \infty} \var \left[ (n\sigma_n)^{\frac{1}{2}} \hat{S}_n^1(t^*,\sigma_n) \right] = \alpha_2 K.
        \end{gather}
    \end{lem}
    \begin{proof}
        First, I will prove that $\lim_{n\rightarrow \infty} \E \left[ \sigma_n^{-h} \hat{S}_n^1(t^*,\sigma_n) \right] = A$:
        \begin{align}
            \E \left[ \sigma_n^{-h} \hat{S}_n^1(t^*,\sigma_n) \right] &= - \frac{\sigma_n^{-h}}{\sigma_n} \E\left[\left(\frac{D Y}{p(X)} - \frac{(1-D) Y}{(1-p(X))} \right) k'\left(\frac{X_i - t^*}{\sigma_n} \right) \right] \\
            &= - \frac{\sigma_n^{-h}}{\sigma_n} \E\left[\left(Y_1 - Y_0 \right) k'\left(\frac{X_i - t^*}{\sigma_n} \right) \right] \\
            &= - \sigma_n^{-h} \int_x \int_y \left(Y_1 - Y_0 \right) \frac{1}{\sigma_n} k'\left(\frac{X_i - t^*}{\sigma_n} \right) \varphi_{y,x}(y,x) d y d x \\
            &= - \sigma_n^{-h} \int_\zeta \int_y \left(Y_1 - Y_0 \right) k'\left(\zeta \right) \varphi_{y,x}(y,t^* + \zeta \sigma_n) d y d \zeta
        \end{align}
        where in the last line I made the substitution $\zeta = \frac{X_i - t^*}{\sigma_n}$. Consider the Taylor expansion of $\varphi$ around $\varphi(y,t^*)$:
        \begin{align}
            \varphi(y,t^* + \zeta \sigma_n) &= \varphi(y,t^*) + \zeta \sigma_n \varphi^{1}_2(y,t^*) + \frac{1}{2} (\zeta \sigma_n)^2 \varphi^{2}_2(y,t^*) + \dots \\
            &= \varphi(y,t^*) + \left( \sum_{i=1}^{h-1} \frac{1}{i!} \varphi^{i}_2(y,x=t^*) \zeta^i \sigma_n^i \right) + \frac{1}{h!} \varphi^{h}_2(y,\tilde{t}) \zeta^h \sigma_n^h
        \end{align}
        with $|\tilde{t}-t^*| \leq |t^* + \zeta \sigma_n - t^*|$. Existence of $\varphi^{m}_2(y,x)$, the $m$-derivatives of $\varphi(y,x)$ with respect to its second argument, is guaranteed by Assumption \ref{ass:smoot} with $s=h+1$.

        Write $\E \left[ \sigma_n^{-h} \hat{S}_n^1(t^*,\sigma_n) \right] $ as $I_{n1} + I_{n2} + I_{n3}$, where:
        \begin{align}
            I_{n1} =&- \sigma_n^{-h} \int_\zeta \int_y \left(Y_1 - Y_0 \right) k'\left(\zeta \right)  \varphi(y,t^*) d y d \zeta \\
            =&  - \sigma_n^{-h}\int_\zeta k'\left(\zeta \right)d \zeta \underbrace{\int_y \left(Y_1 - Y_0 \right) \varphi(y,t^*)
            d y}_{=\E[Y_1 - Y_0 |X=t^*]=0} = 0 \\
            I_{n2} =& - \sigma_n^{-h} \int_\zeta \int_y \left(Y_1 - Y_0 \right) k'\left(\zeta \right)  \left( \sum_{i=1}^{h-1} \frac{1}{i!} \varphi^{i}_2(y,x=t^*) \zeta^i \sigma_n^i \right) d y d \zeta \\
            =& \sigma_n^{-h} \int_y \left(Y_1 - Y_0 \right) \left( \sum_{i=1}^{h-1} \frac{1}{i!} \varphi^{i}_2(y,x=t^*) \sigma_n^i \underbrace{\int_\zeta k'\left(\zeta \right) \zeta^i  d \zeta}_{=0} \right) d y = 0.
        \end{align}
        Result on $I_{n1}$ follows from definition of $t^*$, while Assumption \ref{ass:extraonk} guarantees result on $I_{n2}$. Finally, consider $I_{n3}$:
        \begin{align}
             I_{n3} =&- \sigma_n^{-h} \int_\zeta \int_y \left(Y_1 - Y_0 \right) k'\left(\zeta \right)  \frac{1}{h!} \varphi^{h}_2(y,\tilde{t}) \zeta^h \sigma_n^h d y d \zeta \\
             =& -\frac{1}{h!} \int_\zeta k'\left(\zeta \right) \zeta^h d \zeta \int_y \left(Y_1 - Y_0 \right) \varphi^{h}_2(y,\tilde{t}) d y
        \end{align}
        and conclude that:
        \begin{gather}
            \lim_{n\rightarrow \infty} \E \left[ \sigma_n^{-h} \hat{S}_n^1(t^*,\sigma_n) \right] = -\frac{1}{h!} \int_\zeta \zeta^h k'\left(\zeta \right) d \zeta \int_y \left(Y_1 - Y_0 \right) \varphi^{h}_2(y,t^*) d y = A.
        \end{gather}

        Now I will prove that $\lim_{n\rightarrow \infty} \var \left[ (n\sigma_n)^{\frac{1}{2}} \hat{S}_n^1(t^*,\sigma_n) \right] = \alpha_2 K$. Note that:
        \begin{align}
            \var \left[ (n\sigma_n)^{\frac{1}{2}} \hat{S}_n^1(t^*,\sigma_n) \right] =& \var \left[ (n\sigma_n)^{-\frac{1}{2}} \sum_{i=1}^n \left[\left(\frac{D_i Y_i}{p(X_i)} - \frac{(1-D_i) Y_i}{(1-p(X_i))} \right) k'\left(\frac{X_i - t^*}{\sigma_n} \right) \right] \right] = \\
            =& \sigma_n \var \left[ \frac{1}{\sigma_n} \left(\frac{D Y}{p(X)} - \frac{(1-D) Y}{(1-p(X))} \right) k'\left(\frac{X - t^*}{\sigma_n} \right) \right] = \\
            =& \sigma_n \E \left[ \frac{1}{\sigma_n^2} \left(\frac{D Y}{p(X)} - \frac{(1-D) Y}{(1-p(X))} \right)^2 k'\left(\frac{X - t^*}{\sigma_n} \right)^2 \right] - \\
            & \sigma_n \E \left[ \frac{1}{\sigma_n} \left(\frac{D Y}{p(X)} - \frac{(1-D) Y}{(1-p(X))} \right) k'\left(\frac{X - t^*}{\sigma_n} \right) \right]^2.
        \end{align}
        The second term in the last expression goes to 0 as $n \rightarrow \infty$. For the first term, observe that:
        \begin{gather}
            \sigma_n \E \left[ \frac{1}{\sigma_n} \left(\frac{D Y}{p(X)} - \frac{(1-D) Y}{(1-p(X))} \right)^2 k'\left(\frac{X - t^*}{\sigma_n} \right)^2 \right] = \\
            \int_x \int_y \left(\frac{D Y}{p(X)} - \frac{(1-D) Y}{(1-p(X))} \right)^2 k'\left(\frac{X - t^*}{\sigma_n} \right)^2 \frac{1}{\sigma_n} \varphi_{y,x}(y,x) d y d x = \\
            \int_\zeta \int_y \left(\frac{D Y}{p(X)} - \frac{(1-D) Y}{(1-p(X))} \right)^2 k'\left(\zeta \right)^2 \varphi_{y,x}(y,t^* + \zeta \sigma_n) d y d \zeta
        \end{gather}
        where in the last line I made the substitution $\zeta = \frac{X_i - t^*}{\sigma_n}$. Conclude that
        \begin{align}
            \var \left[ (n\sigma_n)^{\frac{1}{2}} \hat{S}_n^1(t^*,\sigma_n) \right] =& \int_\zeta k'\left(\zeta \right)^2 d \zeta f_x(t^*) \E \left [ \left(\frac{D Y}{p(X)} - \frac{(1-D) Y}{(1-p(X))} \right)^2 | X 
            = t^* \right] \\ =& \alpha_2 K.
        \end{align}
        Note that $\alpha_2 K$ is bounded by Assumptions \ref{ass:sqintegr}, \ref{ass:smoot}, and \ref{ass:extraonk}.
    \end{proof}

\subsubsection*{Lemma 2}
    \begin{lem}\label{lemma2}
        Let Assumptions \ref{ass:identification}, \ref{ass:consistency} (with $s=h + 1$ for some $h\geq 2$), and \ref{ass:asymptotic_sm} hold. If $n \sigma_n^{2h + 1}\rightarrow \infty$, $\sigma_n^{-h} \hat{S}_n^1(t^*,\sigma_n)$ converges in probability to A. If $n \sigma_n^{2h + 1}$ has a finite limit $\lambda$, $(n\sigma_n)^{\frac{1}{2}} \hat{S}_n^1(t^*,\sigma_n)$ converges in distribution to $\mathcal{N}(\lambda^{\frac{1}{2}}A,\alpha_2 K)$.
    \end{lem}
    \begin{proof}
    Note that $$\var \left[ (\sigma_n)^{-h} \hat{S}_n^1(t^*,\sigma_n) \right] = \underbrace{(n \sigma_n^{2h + 1})^{-1}}_{\rightarrow 0} \underbrace{\var \left[ (n\sigma_n)^{\frac{1}{2}} \hat{S}_n^1(t^*,\sigma_n) \right]}_{\rightarrow \alpha_2 K}.$$ So the first result follows from lemma \ref{lemma1} and Chebyshev's inequality.

    For the second result, first note that under the stated assumptions and from lemma \ref{lemma1}, $$ \E \left[ (n\sigma_n)^{\frac{1}{2}} \hat{S}_n^1(t^*,\sigma_n) \right] = \underbrace{(n \sigma_n^{2h +1})^{\frac{1}{2}}}_{\rightarrow \lambda^{\frac{1}{2}}} \underbrace{\E \left[ \sigma_n^{-h} \hat{S}_n^1(t^*,\sigma_n) \right]}_{A} $$ and so the result follows if I show that $$ U_n = (n\sigma_n)^{\frac{1}{2}} \left( \hat{S}_n^1(t^*,\sigma_n) - \E \left[ \hat{S}_n^1(t^*,\sigma_n)\right] \right) \rightarrow^d \mathcal{N}(0,\alpha_2 K).$$
    Note that 
    \begin{align}
         U_n =& (n\sigma_n)^{\frac{1}{2}} \frac{1}{n} \sum_{i=1}^n \left[ \underbrace{\left(\frac{D_i Y_i}{p(X_i)} - \frac{(1-D_i) Y_i}{(1-p(X_i))} \right) k'\left(\frac{X_i - t}{\sigma_n} \right) \frac{1}{\sigma_n}}_{=B} - \E[B] \right] =\\
         =&  \sum_{i=1}^n \left(\frac{\sigma_n}{n}\right)^{\frac{1}{2}}(B - \E[B])
    \end{align}
    and hence $U_n$ has characteristic function $\psi(\tau)^n$, where $$ \psi(\tau) = \E \left[ \exp\left(i \tau \left(\frac{\sigma_n}{n}\right)^{\frac{1}{2}}(B - \E[B])\right) \right]$$ and
    \begin{gather}
        \psi^{'}(\tau) = \E \left[ i \left(\frac{\sigma_n}{n}\right)^{\frac{1}{2}}(B - \E[B]) \exp\left(i \tau \left(\frac{\sigma_n}{n}\right)^{\frac{1}{2}}(B - \E[B])\right) \right] \\
        \psi^{''}(\tau) = \E \left[ - \frac{\sigma_n}{n}(B - \E[B]) \exp\left(i \tau \left(\frac{\sigma_n}{n}\right)^{\frac{1}{2}}(B - \E[B])\right) \right].
    \end{gather}
    Note that $\psi^{'}(0)= 0$ and $\psi^{''}(0)= - \frac{\sigma_n}{n} \var(B) = - \frac{1}{n}(\alpha_2 K + o(1))$, since lemma \ref{lemma1} proved that $\lim_{n \rightarrow \infty} \sigma_n \var(B) = \alpha_2 K$.
    
    A Taylor series expansion of $\psi(\tau)$ about $\tau = 0$ yields: $$\psi(\tau) = \underbrace{\psi(0)}_{=1} + \underbrace{\psi^{'}(0)}_{=0}\tau + \frac{1}{2} \underbrace{\psi^{''}(0)}_{= -\frac{\alpha_2 K}{n}}\tau^2 + o \left(\frac{\tau^2}{n}\right) = 1 - \frac{1}{2n}\alpha_2 K \tau^2 + o \left(\frac{\tau^2}{n}\right)$$
    and hence the characteristic function of $U_n$ has limit:
    \begin{gather}
        \lim_{n \rightarrow \infty} \left [ 1 - \frac{1}{2n}\alpha_2 K \tau^2 + o \left(\frac{\tau^2}{n}\right) \right]^n = \exp \left( -\alpha_2 K \frac{\tau^2}{2} \right).
    \end{gather}
    Since $\exp \left( -\alpha_2 K \frac{\tau^2}{2} \right)$ is the characteristic function of $\mathcal{N}(0,\alpha_2 K)$, the second result of the lemma holds.
    \end{proof}

\subsubsection*{Lemma 3}
    \begin{lem}\label{lemma3}
       Let Assumptions \ref{ass:identification}, \ref{ass:consistency} (with $s=h + 1$ for some $h\geq 2$), and \ref{ass:asymptotic_sm} hold. Let $\eta > 0$ be such that $\varphi_{y,x}(y,x)$ has second derivative uniformly bounded for almost every $X$ if $|X-t^*|<\eta$. For $\theta\in \R$, define $\hat{S}^\theta_n(\theta)$ by
        \begin{gather}
            \hat{S}^\theta_n(\theta) = -(n \sigma_n^2)^{-1} \sum_{i=1}^n \left[\left(\frac{D_i Y_i}{p(X_i)} - \frac{(1-D_i) Y_i}{(1-p(X_i))} \right) k'\left(\frac{X_i - t^*}{\sigma_n} + \theta \right) \right].
        \end{gather}
        Define the sets $\Theta_n(n=1,2,\dots)$ by $\Theta_n = \{\theta: \theta \in \R, \sigma_n |\theta| \leq \frac{\eta}{2}\}$.
        Then
        \begin{gather}
            \plim_{n \rightarrow \infty}  \sup_{\theta\in \Theta_n} |\hat{S}^\theta_n(\theta) - \E[\hat{S}^\theta_n(\theta)]| = 0.
        \end{gather}
        In addition, there are finite numbers $a_1$ and $a_2$ such that for all $\theta \in \Theta_n$
        $$|\E[\hat{S}^\theta_n(\theta)] - H \theta| \leq o(1) + a_1 \sigma_n |\theta| + a_2 \sigma_n \theta^2 $$
        uniformly over $\theta \in \Theta_n$.
    \end{lem}
    \begin{proof}
        To prove the first result, first define 
        \begin{align}
            -g_i(\theta) =& \left(\frac{D_i Y_i}{p(X_i)} - \frac{(1-D_i) Y_i}{(1-p(X_i))} \right) k'\left(\frac{X_i - t^*}{\sigma_n} + \theta \right) - \\
            & \E \left[ \left(\frac{D_i Y_i}{p(X_i)} - \frac{(1-D_i) Y_i}{(1-p(X_i))} \right) k'\left(\frac{X_i - t^*}{\sigma_n} + \theta \right) \right].
        \end{align}
        It is necessary to prove that for any $\varepsilon>0$
        $$
        \lim _{n \rightarrow \infty} \operatorname{Pr}\left[\sup _{\theta \in \Theta_n}\left(n \sigma_n^2\right)^{-1}\left|\sum_{n=1}^N g_{i}(\theta)\right|>\varepsilon\right]=0 .
        $$
        Given any $\delta>0$, divide each set $\Theta_n$ into nonoverlapping subsets $\Theta_{n j}(j=1,2, \ldots)$ such that the distance between any two points in the same subset does not exceed $\delta \sigma_n^2$ and the number $\Gamma_n$ of subsets does not exceed $C \sigma_n^{-3(q-1)}$ for some $C>0$. Let $\left\{\theta_{N i}\right\}$ be a set of vectors such that $\theta_{nj} \in \Theta_{nj}$. Then
        \begin{align}
        \operatorname{Pr}\left[ \sup _{\theta \in \Theta_n}\left(n \sigma_n^2\right)^{-1} \right. & \left. \left|\sum_{n=1}^n g_{i}(\theta)\right|>\varepsilon\right] = \\
        = & \operatorname{Pr}\left[\bigcup_{j=1}^{\Gamma_n} \sup _{\theta \in \Theta_{nj}}\left(n \sigma_n^2\right)^{-1}\left|\sum_{i=1}^n g_{i}(\theta)\right|>\varepsilon\right] \\
        \leqslant & \sum_{j=1}^{\Gamma_n} \operatorname{Pr}\left[\sup _{\theta \in \Theta_{nj}}\left(n \sigma_n^2\right)^{-1}\left|\sum_{i=1}^n g_{i}(\theta)\right|>\varepsilon\right] \\
        \leqslant & \underbrace{\sum_{j=1}^{\Gamma_n} \operatorname{Pr}\left[\left(n \sigma_n^2\right)^{-1}\left|\sum_{i=1}^n g_{i}\left(\theta_{nj}\right)\right|>\varepsilon / 2\right]}_{=B_1} \\
        & + \underbrace{ \sum_{j=1}^{\Gamma_n} \operatorname{Pr}\left[\left(n \sigma_n^2\right)^{-1} \sum_{i=1}^n \sup _{\theta \in \Theta_{nj}}\left|g_{i}(\theta)-g_{i}\left(\theta_{nj}\right)\right|>\varepsilon / 2\right]}_{=B_2},
        \end{align}
        where the last two lines follow from the triangle inequality. By Hoeffding's inequality, there are finite numbers $c_1>0$ and $c_2>0$ such that
        $$
        \operatorname{Pr}\left[\left(n \sigma_n^2\right)^{-1}\left|\sum_{i=1}^n g_{i}\left(\theta_{nj}\right)\right|>\varepsilon / 2\right] \leqslant c_1 \exp \left(-c_2 n \sigma_n^4\right) .
        $$
        Therefore, $B_1$ is bounded by $C c_1 \sigma_n^{-3(q-1)} \exp \left(-c_2 n \sigma_n^4\right)$, which converges to 0 as $n \rightarrow \infty$ by Assumption \ref{ass:rateofs}. In addition, by Assumption \ref{ass:kernel} there is a finite $c_3>0$ such that if $\theta \in \Theta_{nj}$,
        \begin{gather}
        \left| -\left(\frac{D_i Y_i}{p(X_i)} - \frac{(1-D_i) Y_i}{(1-p(X_i))} \right) \left[k'\left(\frac{X_i - t^*}{\sigma_n} + \theta \right) - k'\left(\frac{X_i - t^*}{\sigma_n} + \theta_{nj} \right) \right]\right| \\
        \leqslant c_3\left|\theta-\theta_{nj}\right| \leqslant c_3 \delta \sigma_n^2 .
        \end{gather}
        So
        $$
        \left(n \sigma_n^2\right)^{-1} \sum_{i=1}^n \sup _{\theta \in \Theta_{n j}}\left|g_{i}(\theta)-g_{i}\left(\theta_{nj}\right)\right| \leqslant 2 c_3 \delta .
        $$
        Choose $\delta<\varepsilon / 4 c_3$. Then $B_2$ is 0. This establishes $\plim_{n \rightarrow \infty}  \sup_{\theta\in \Theta_n} |\hat{S}^\theta_n(\theta) - \E[\hat{S}^\theta_n(\theta)]| = 0$.

        To prove the second result, start noting that
        \begin{gather}
            \E[\hat{S}^\theta_n(\theta)] = -\sigma_n^{-2} \E \left[ \left(\frac{D Y}{p(X)} - \frac{(1-D) Y}{(1-p(X))} \right) k'\left(\frac{X - t^*}{\sigma_n} + \theta \right) \right] = I_{n1} + I_{n2}
        \end{gather}
        where
        \begin{gather}
            I_{n1} = -\sigma_n^{-2} \int_{|X-t^*| \leq \eta} \int_y \left(\frac{D Y}{p(X)} - \frac{(1-D) Y}{(1-p(X))} \right) k'\left(\frac{X - t^*}{\sigma_n} + \theta \right) \varphi(y,x) dy dx \\
            I_{n2} = -\sigma_n^{-2} \int_{|X-t^*|>\eta} \int_y \left(\frac{D Y}{p(X)} - \frac{(1-D) Y}{(1-p(X))} \right) k'\left(\frac{X - t^*}{\sigma_n} + \theta \right) \varphi(y,x) dy dx.
        \end{gather}
        First, consider $I_{n2}$ and observe that
        \begin{gather}
            I_{n2} = -\sigma_n^{-2} \int_{|X-t^*|>\eta} k'\left(\frac{X - t^*}{\sigma_n} + \theta \right) \underbrace{\int_y \left(\frac{D Y}{p(X)} - \frac{(1-D) Y}{(1-p(X))} \right)  \varphi(y|x) dy}_{= \E[Y_1 - Y_0 |X]} f_x(x) dx
        \end{gather}
        and since $\E[Y_1 - Y_0 |X]$ is bounded by Assumption \ref{ass:sqintegr}, $$|I_{n2}| \leq \left|C \sigma_n^{-2} \int_{|X-t^*|>\eta} k'\left(\frac{X - t^*}{\sigma_n} + \theta \right) f_x(x) dx \right|.$$ Define $\zeta = \frac{X - t^*}{\sigma_n} + \theta$. Since $\sigma_n |\theta| \leq \frac{\eta}{2}$, when $|X-t^*| > \eta$
        \begin{align}
            |\zeta| =& \left| \frac{X - t^*}{\sigma_n} + \theta \right| = \pm  \left( \frac{X - t^*}{\sigma_n} + \theta \right) \geq \pm  \left( \frac{X - t^*}{\sigma_n}\right) - |\theta | = \left| \frac{X - t^*}{\sigma_n}\right| - |\theta |  \\
            \geq& \left| \frac{X - t^*}{\sigma_n}\right| - \frac{\eta}{2 \sigma_n} > \frac{\eta}{\sigma_n}- \frac{\eta}{2 \sigma_n} = \frac{\eta}{2 \sigma_n}.
        \end{align}
        and so the event $|X-t^*| > \eta$ implies $|\zeta| > \frac{\eta}{2 \sigma_n}$. Then
        \begin{align}
            |I_{n2}| \leq & \left| C \sigma_n^{-2} \int_{|X-t^*|>\eta} k'\left(\frac{X - t^*}{\sigma_n} + \theta \right) f_x(x) dx \right| \\
            =& \left|C \sigma_n^{-1} \int_{|X-t^*|>\eta} k'\left(\zeta \right) f_x(t^* - \theta\sigma_n) d\zeta \right| \\
            \leq & \left| C \underbrace{f_x(t^* - \theta\sigma_n)}_{\rightarrow f_x(t^*)} \underbrace{\sigma_n^{-1}  \int_{|\zeta|>\eta/\sigma_n} k'\left(\zeta \right) d\zeta}_{\rightarrow 0} \right|.
        \end{align}
        The fact that $f_x(t^* - \theta\sigma_n)\rightarrow f_x(t^*)$ bounded by Assumption \ref{ass:smoot} with $s=h+1$ and $\sigma_n^{-1}  \int_{|\zeta|>\eta/\sigma_n} k'\left(\zeta \right) d\zeta \rightarrow 0$ by Assumption \ref{ass:extraonk} implies $$\plim_{n \rightarrow \infty}  \sup_{\theta\in \Theta_n} |I_{n2}| = 0.$$

        Recall that $I_{n1}$ is defined as $$I_{n1} = - \sigma_n^{-2} \int_{|X-t^*| \leq \eta} \int_y \left(\frac{D Y}{p(X)} - \frac{(1-D) Y}{(1-p(X))} \right) k'\left(\frac{X - t^*}{\sigma_n} + \theta \right) \varphi(y,x) dy dx.$$ Consider a Taylor expansion of $\varphi(y,x)$ about $x=t^*$: $$ \varphi(y,x) = \varphi(y,t^*) + \varphi'(y,t^*)(x - t^*) + \frac{1}{2} \varphi''(y,\tilde{t})(x - t^*)^2$$ with $|\tilde{t}-t^*| \leq |x - t^*|$. Write $I_{n1}$ as $J_{n1}+J_{n2}+J_{n3}$ where
        \begin{align}
            J_{n1} =& -\sigma_n^{-2} \int_{|X-t^*| \leq \eta} \int_y \left(\frac{D Y}{p(X)} - \frac{(1-D) Y}{(1-p(X))} \right) k'\left(\frac{X - t^*}{\sigma_n} + \theta \right) \varphi(y,t^*) dy dx \\
            =& -\sigma_n^{-2} \underbrace{\E[Y_1 -Y_0|X=t^*]}_{=0} \int_{|X-t^*| \leq \eta} k'\left(\frac{X - t^*}{\sigma_n} + \theta \right) f_x(t^*) dx = 0 \\
            J_{n2} =& -\sigma_n^{-2} \int_{|X-t^*| \leq \eta} \int_y \left(\frac{D Y}{p(X)} - \frac{(1-D) Y}{(1-p(X))} \right) k'\left(\frac{X - t^*}{\sigma_n} + \theta \right) \varphi'(y,t^*)(x - t^*) dy dx \\
            J_{n3} =& -\sigma_n^{-2} \int_{|X-t^*| \leq \eta} \int_y \left(\frac{D Y}{p(X)} - \frac{(1-D) Y}{(1-p(X))} \right) k'\left(\frac{X - t^*}{\sigma_n} + \theta \right) \frac{1}{2} \varphi''(y,\tilde{t})(x - t^*)^2 dy dx.
        \end{align}
        Consider $J_{n2}$ and the substitution $\zeta = \frac{X - t^*}{\sigma_n} + \theta$:
        \begin{align}
            J_{n2} =& -\int_{|\zeta - \theta| \leq \eta/\sigma_n} \int_y \left(\frac{D Y}{p(X)} - \frac{(1-D) Y}{(1-p(X))} \right) k'\left(\zeta \right) (\zeta - \theta) \varphi'(y,t^*) dy d\zeta \\
            =& -\int_{|\zeta - \theta| \leq \eta/\sigma_n} \int_y \left(\frac{D Y}{p(X)} - \frac{(1-D) Y}{(1-p(X))} \right) k'\left(\zeta \right) \zeta \varphi'(y,t^*) dy d\zeta + \\
            & \int_{|\zeta - \theta| \leq \eta/\sigma_n} \int_y \left(\frac{D Y}{p(X)} - \frac{(1-D) Y}{(1-p(X))} \right) k'\left(\zeta \right)  \theta \varphi'(y,t^*) dy d\zeta \\
            =& -\int_{|\zeta - \theta| \leq \eta/\sigma_n} \zeta k'\left(\zeta \right) d\zeta \underbrace{\int_y \left(\frac{D Y}{p(X)} - \frac{(1-D) Y}{(1-p(X))} \right)  \varphi'(y,t^*) dy}_{=H} + \\
            & \theta \int_{|\zeta - \theta| \leq \eta/\sigma_n} k'\left(\zeta \right) d\zeta \underbrace{\int_y \left(\frac{D Y}{p(X)} - \frac{(1-D) Y}{(1-p(X))} \right)  \varphi'(y,t^*) dy}_{=H}
        \end{align}
        where $H=  f_x(t^*) \left(\frac{\partial \E_P\left[Y_1 - Y_0 | X = t^* \right]}{\partial X} \right)$ is bounded by Assumption \ref{ass:smoot}. Since $\int \zeta k'\left(\zeta \right) d\zeta = 0$ and $\sigma_n|\theta| \leq \frac{\eta}{2}$:
        \begin{gather}
            \left| \int_{|\zeta - \theta| \leq \eta/\sigma_n} \zeta k'\left(\zeta \right) d\zeta \right| = \left| \int_{|\zeta - \theta| > \eta/\sigma_n} \zeta k'\left(\zeta \right) d\zeta \right| \leq \left| \int_{|\zeta | > \eta/2\sigma_n} \zeta k'\left(\zeta \right) d\zeta \right|.
        \end{gather}
        By Assumption \ref{ass:extraonk}, $\left| \int_{|\zeta | > \eta/2\sigma_n} \zeta k'\left(\zeta \right) d\zeta \right|$ converges to 0 uniformly over $\theta \in \Theta_n$. It means that $\int_{|\zeta - \theta| \leq \eta/\sigma_n} \zeta k'\left(\zeta \right) d\zeta$ converges uniformly to 0.
        
        Consider $\theta H \int_{|\zeta - \theta| \leq \eta/\sigma_n} k'\left(\zeta \right) d\zeta$, and note that, since $\int k'\left(\zeta \right) d\zeta = 1$,
        \begin{gather}
            \left|\theta H - \theta H \int_{|\zeta - \theta| \leq \eta/\sigma_n} k'\left(\zeta \right) d\zeta \right| =  \left| \theta H \int_{|\zeta - \theta| > \eta/\sigma_n} k'\left(\zeta \right) d\zeta \right| \leq \\
            |\sigma_n \theta H | \sigma_n^{-1} \int_{|\zeta - \theta| > \eta/\sigma_n} k'\left(\zeta \right) d\zeta \leq
            \frac{\eta}{2} \sigma_n^{-1} \int_{|\zeta - \theta| > \eta/\sigma_n} k'\left(\zeta \right) d\zeta.
        \end{gather}
        The last term is bounded uniformly over $n$ and $\theta \in \Theta_n$ and converges to 0 by Assumption \ref{ass:extraonk}. It means that
        \begin{gather}
            \lim_{n \rightarrow \infty} \left| \sup_{\theta \in \Theta_n} J_{n2} - \theta H \right| = 0.
        \end{gather}
        Finally, consider $J_{n3}$:
        \begin{align}
            |J_{n3}| =& \left|\frac{1}{2}\sigma_n^{-2} \int_{|X-t^*| \leq \eta} \int_y \left(\frac{D Y}{p(X)} - \frac{(1-D) Y}{(1-p(X))} \right) k'\left(\frac{X - t^*}{\sigma_n} + \theta \right)  \varphi''(y,\tilde{t})(x - t^*)^2 dy dx \right| \\
            =& \left| \frac{1}{2}\sigma_n \int_{|\zeta - \theta| \leq \eta/\sigma_n} \int_y \left(\frac{D Y}{p(X)} - \frac{(1-D) Y}{(1-p(X))} \right) k'\left(\zeta \right)  \varphi''(y,\tilde{t})(\zeta - \theta)^2 dy dx \right| \\
            \leq& \underbrace{\left| \frac{1}{2}\sigma_n  \int_{|\zeta - \theta| \leq \eta/\sigma_n} \int_y \left(\frac{D Y}{p(X)} - \frac{(1-D) Y}{(1-p(X))} \right) k'\left(\zeta \right) \zeta^2 \varphi''(y,\tilde{t}) dy dx \right|}_{=o(1)} + \\
            & \sigma_n |\theta| \underbrace{\left|  \int_{|\zeta - \theta| \leq \eta/\sigma_n} \int_y \left(\frac{D Y}{p(X)} - \frac{(1-D) Y}{(1-p(X))} \right) k'\left(\zeta \right) \zeta  \varphi''(y,\tilde{t}) dy dx \right|}_{=a_1} + \\
            & \sigma_n |\theta|^2 \underbrace{\left|  \int_{|\zeta - \theta| \leq \eta/\sigma_n} \int_y \left(\frac{D Y}{p(X)} - \frac{(1-D) Y}{(1-p(X))} \right) k'\left(\zeta \right)  \varphi''(y,\tilde{t}) dy dx \right|}_{=a_2}.
        \end{align}
        Combine results on $I_{n2}$ and $J_{n1}$, $J_{n2}$, and $J_{n3}$ to get
        $$|\E[\hat{S}^\theta_n(\theta)] - H \theta| = |J_{n1} + J_{n2} + J_{n3} + I_{n2} - H \theta| \leq o(1) + a_1 \sigma_n |\theta| + a_2 \sigma_n \theta^2 $$
        uniformly over $\theta \in \Theta_n$, which proves the second part of the lemma.
    \end{proof}

\subsubsection*{Lemma 4}
    \begin{lem}\label{lemma4}
         Let Assumptions \ref{ass:identification}, \ref{ass:consistency} (with $s=h + 1$ for some $h\geq 2$), and \ref{ass:asymptotic_sm} hold. Define $\hat{\theta}_n = \frac{t^* - \hat{t}^s_n}{\sigma_n}$. Then $\plim_{n \rightarrow \infty} \hat{\theta}_n = 0$.
    \end{lem}
    \begin{proof}
        Consider $\hat{S}^\theta_n(\hat{\theta}_n)$:
        \begin{align}
            \hat{S}^\theta_n(\hat{\theta}_n) =& - (n \sigma_n^2)^{-1} \sum_{i=1}^n \left[\left(\frac{D_i Y_i}{p(X_i)} - \frac{(1-D_i) Y_i}{(1-p(X_i))} \right) k'\left(\frac{X_i - t^*}{\sigma_n} + \hat{\theta}_n \right) \right] \\
            =& -(n \sigma_n^2)^{-1} \sum_{i=1}^n \left[\left(\frac{D_i Y_i}{p(X_i)} - \frac{(1-D_i) Y_i}{(1-p(X_i))} \right) k'\left(\frac{X_i - \hat{t}^s_n}{\sigma_n} \right) \right].
        \end{align}
        By Theorem \ref{thm:cons_sm}, $\hat{t}^s_n \rightarrow^{a.s.} t^*$. With probability approaching 1, then, $\hat{t}^s_n$ is an interior point of $\mathcal{T}$. It means that, with probability approaching 1, $\hat{S}^\theta_n(\hat{\theta}_n) = \hat{S}^1_n(\hat{t}^s_n,\sigma_n) = 0$. Hence lemma \ref{lemma3} gives
        \begin{gather}
            |H\hat{\theta}_n| \leq o(1) + a_1 \sigma_n |\hat{\theta}_n| + a_2 \sigma_n \hat{\theta}_n^2
        \end{gather}
        with $H \neq 0$ by Assumptions \ref{ass:uniq} and \ref{ass:smoot}.
        
        I will hence prove $\plim_{n \rightarrow \infty} \hat{\theta}_n = 0$ by contradiction. First, assume that $\hat{\theta}_n$ has finite limit different from 0. The left-hand side of the previous inequality would be positive, while the right-hand side converges to 0. This contradicts the inequality. Then assume the limit is unbounded. By Theorem \ref{thm:cons_sm}, $\plim_{n \rightarrow \infty} \sigma_n \hat{\theta}_n = 0$. This gives the contradiction
        \begin{gather}
            \underbrace{\frac{|H\hat{\theta}_n|}{|\hat{\theta}_n|}}_{=|H|>0} \leq \underbrace{\frac{ o(1)}{|\hat{\theta}_n|}}_{\rightarrow 0} + \underbrace{\frac{a_1 \sigma_n |\hat{\theta}_n|}{|\hat{\theta}_n|}}_{=a_1 \sigma_n\rightarrow 0} + \underbrace{\frac{a_2 \sigma_n \hat{\theta}_n^2}{|\hat{\theta}_n|}}_{=a_2 \sigma_n|\hat{\theta}_n|\rightarrow 0}
        \end{gather}
        and proves that $\plim_{n \rightarrow \infty} \hat{\theta}_n = 0$.
    \end{proof}
    
\subsubsection*{Lemma 5}
    \begin{lem}\label{lemma5}
        Let Assumptions \ref{ass:identification}, \ref{ass:consistency} (with $s=h + 1$ for some $h\geq 2$), and \ref{ass:asymptotic_sm} hold. Consider a sequence $t_n$ such that $\frac{t_n - t^*}{\sigma_n} \rightarrow 0$. Then
        \begin{gather}
            \plim_{n \rightarrow \infty} \hat{S}^2_n(t_n,\sigma_n) = -H.
        \end{gather}
    \end{lem}
    \begin{proof}
        To prove the lemma it is sufficient to show that $\E[\hat{S}^2_n(t_n,\sigma_n)] \rightarrow H$ and $\var(\hat{S}^2_n(t_n,\sigma_n)) \rightarrow 0$. Recall that $$\hat{S}_n^2(t,\sigma_n) = \frac{1}{\sigma_n^2} \frac{1}{n} \sum_{i=1}^n \left[\left(\frac{D_i Y_i}{p(X_i)} - \frac{(1-D_i) Y_i}{(1-p(X_i))} \right) k''\left(\frac{X_i - t}{\sigma_n} \right) \right]$$ and hence 
        \begin{align}
            \E[\hat{S}^2_n(t_n,\sigma_n)] =& \frac{1}{\sigma_n^2} \E\left[ \left(\frac{D Y}{p(X)} - \frac{(1-D) Y}{1-p(X)} \right) k''\left(\frac{X - t_n}{\sigma_n} \right) \right] \\
            =& \frac{1}{\sigma_n^2} \int_x \int_y \left(\frac{D Y}{p(X)} - \frac{(1-D) Y}{1-p(X)} \right) k''\left(\frac{X - t_n}{\sigma_n} \right) \varphi(y,x) dy dx.
        \end{align}
        Consider a Taylor expansion of $\varphi(y,x)$ about $x=t^*$: $$ \varphi(y,x) = \varphi(y,t^*) + \varphi'(y,t^*)(x - t^*) + \frac{1}{2} \varphi''(y,\tilde{t})(x - t^*)^2$$ with $|\tilde{t}-t^*| \leq |x - t^*|$. Let $\eta > 0$ be such that $\varphi_{y,x}(y,x)$ has second derivative uniformly bounded for almost every $X$ if $|X-t^*|<\eta$, and write $\E[\hat{S}^2_n(t_n,\sigma_n)]$ as $I_{n1}+I_{n2}+I_{n3}+I_{n4}$, where:
        \begin{align}
            I_{n1} =& \sigma_n^{-2} \int_{|X-t^*| \leq \eta} \int_y \left(\frac{D Y}{p(X)} - \frac{(1-D) Y}{(1-p(X))} \right) k''\left(\frac{X - t_n}{\sigma_n} \right) \varphi(y,t^*) dy dx \\
            =& \sigma_n^{-2} \underbrace{\E[Y_1 -Y_0|X=t^*]}_{=0} \int_{|X-t^*| \leq \eta} k''\left(\frac{X - t_n}{\sigma_n} \right) f_x(t^*) dx = 0 \\
            I_{n2} =& \sigma_n^{-2} \int_{|X-t^*| \leq \eta} \int_y \left(\frac{D Y}{p(X)} - \frac{(1-D) Y}{(1-p(X))} \right) k''\left(\frac{X - t_n}{\sigma_n} \right) \varphi'(y,t^*)(x - t^*) dy dx \\
            I_{n3} =& \sigma_n^{-2} \int_{|X-t^*| \leq \eta} \int_y \left(\frac{D Y}{p(X)} - \frac{(1-D) Y}{(1-p(X))} \right) k''\left(\frac{X - t_n}{\sigma_n} \right) \frac{1}{2} \varphi''(y,\tilde{t})(x - t^*)^2 dy dx \\
            I_{n4} =& \frac{1}{\sigma_n^2} \int_{|X-t^*| > \eta} \int_y \left(\frac{D Y}{p(X)} - \frac{(1-D) Y}{1-p(X)} \right) k''\left(\frac{X - t_n}{\sigma_n} \right) \varphi(y,x) dy dx.
        \end{align}
        Consider the substitution $\zeta = \frac{X - t^*}{\sigma_n} + \theta_n = \frac{X - t_n}{\sigma_n}$:
        \begin{align}
            I_{n2} =& \int_{|\zeta - \theta_n| \leq \eta/\sigma_n} \int_y \left(\frac{D Y}{p(X)} - \frac{(1-D) Y}{(1-p(X))} \right)
            k''\left(\zeta \right) \varphi'(y,t^*)(\zeta - \theta_n) dy d\zeta \\
            =& \int_{|\zeta - \theta_n| \leq \eta/\sigma_n} k''\left(\zeta \right) (\zeta - \theta_n) d\zeta \underbrace{\int_y \left(\frac{D Y}{p(X)} - \frac{(1-D) Y}{(1-p(X))} \right) \varphi'(y,t^*) dy}_{=H} \\
            =& H \int_{|\zeta - \theta_n| \leq \eta/\sigma_n} \zeta k''\left(\zeta \right)d\zeta - \theta_n H \int_{|\zeta - \theta_n| \leq \eta/\sigma_n} k''\left(\zeta \right) d\zeta\\
            =& H \left( \int \zeta k''\left(\zeta \right)d\zeta - \int_{|\zeta - \theta_n| > \eta/\sigma_n} \zeta k''\left(\zeta \right)d\zeta\right) - \theta_n H \int_{|\zeta - \theta_n| \leq \eta/\sigma_n} k''\left(\zeta \right) d\zeta
        \end{align}
        Under Assumption \ref{ass:extraonk}, $\int \zeta k''\left(\zeta \right)d\zeta = -1$, $\int_{|\zeta - \theta_n| > \eta/\sigma_n} \zeta k''\left(\zeta \right)d\zeta \rightarrow^p 0$ and $\int_{|\zeta - \theta_n| \leq \eta/\sigma_n} k''\left(\zeta \right) d\zeta$ is bounded. Since $\theta_n \rightarrow^p 0$, $I_{n2}\rightarrow^p -H$.

        Consider $I_{n3}$ and the substitution $\zeta = \frac{X - t^*}{\sigma_n} + \theta_n = \frac{X - t_n}{\sigma_n}$:
        \begin{align}
            I_{n3} =& \frac{1}{2} \sigma_n \int_y \left(\frac{D Y}{p(X)} - \frac{(1-D) Y}{(1-p(X))} \right) \varphi''(y,\tilde{t}) dy \int_{|\zeta - \theta_n| \leq \eta/\sigma_n} k''\left(\zeta \right)(\zeta - \theta_n)^2 d\zeta.
        \end{align}
        Integrals are bounded by Assumptions \ref{ass:smoot} with $s=h+1$ and \ref{ass:extraonk}, and hence $I_{n3}\rightarrow 0$.

        Finally, consider $I_{n4}$ and the substitution $\zeta = \frac{X - t^*}{\sigma_n} + \theta_n = \frac{X - t_n}{\sigma_n}$:
        \begin{align}
            I_{n4} =& \frac{1}{\sigma_n} \int_{|\zeta - \theta_n| > \eta/\sigma_n} \int_y \left(\frac{D Y}{p(X)} - \frac{(1-D) Y}{1-p(X)} \right) k''\left(\zeta \right) \varphi(y,t^* + \sigma_n(\zeta - \theta_n)) dy d\zeta \\
            \rightarrow^p& \int_y \left(\frac{D Y}{p(X)} - \frac{(1-D) Y}{1-p(X)} \right) \varphi(y,t^* ) dy \frac{1}{\sigma_n} \int_{|\zeta - \theta_n| > \eta/\sigma_n} k''\left(\zeta \right) d\zeta.
        \end{align}
        $\int_y \left(\frac{D Y}{p(X)} - \frac{(1-D) Y}{1-p(X)} \right) \varphi(y,t^* ) dy$ is bounded by Assumptions \ref{ass:sqintegr} and $\frac{1}{\sigma_n} \int_{|\zeta - \theta_n| > \eta/\sigma_n} k''\left(\zeta \right) d\zeta \rightarrow^p 0$ by Assumptions \ref{ass:extraonk}, and hence $I_{n4} \rightarrow^p 0$.

        Combine results on $I_{n1}$, $I_{n2}$, $I_{n3}$, and $I_{n4}$ to conclude $\E[\hat{S}^2_n(t_n,\sigma_n)] \rightarrow H$. Consider now the variance:
        \begin{align}
            \var(\hat{S}^2_n(t_n,\sigma_n)) =& \var\left( \frac{1}{\sigma_n^2} \frac{1}{n} \sum_{i=1}^n \left[\left(\frac{D_i Y_i}{p(X_i)} - \frac{(1-D_i) Y_i}{(1-p(X_i))} \right) k''\left(\frac{X_i - t_n}{\sigma_n} \right) \right] \right) \\
            =& \frac{1}{n \sigma_n^4} \var\left( \left(\frac{D Y}{p(X)} - \frac{(1-D) Y}{(1-p(X))} \right) k''\left(\frac{X - t_n}{\sigma_n} \right) \right) \\
            =& \frac{1}{n \sigma_n^4} \E\left[ \left(\frac{D Y}{p(X)} - \frac{(1-D) Y}{(1-p(X))} \right)^2 k''\left(\frac{X - t_n}{\sigma_n} \right)^2 \right] -\\
            & \frac{1}{n} \left(\underbrace{\frac{1}{\sigma_n}\E\left[ \left(\frac{D Y}{p(X)} - \frac{(1-D) Y}{(1-p(X))} \right) k''\left(\frac{X - t_n}{\sigma_n} \right) \right]}_{\rightarrow^p 0} \right)^2
        \end{align}
        and the substitution $\zeta = \frac{X - t^*}{\sigma_n} + \theta_n = \frac{X - t_n}{\sigma_n}$:
        \begin{gather}
            \frac{1}{n \sigma_n^4} \E\left[ \left(\frac{D Y}{p(X)} - \frac{(1-D) Y}{(1-p(X))} \right)^2 k''\left(\frac{X - t_n}{\sigma_n} \right)^2 \right] = \\
            \frac{1}{n \sigma_n^3} \int_\zeta \int_y \left(\frac{D Y}{p(X)} - \frac{(1-D) Y}{(1-p(X))} \right)^2 k''\left(\zeta \right)^2 \varphi(y,t^* + \sigma_n(\zeta - \theta_n)) dy d\zeta \rightarrow^p \\
            \frac{1}{n \sigma_n^3} \int_y \left(\frac{D Y}{p(X)} - \frac{(1-D) Y}{(1-p(X))} \right)^2 \varphi(y,t^*) dy \int_\zeta k''\left(\zeta \right)^2 d\zeta.
        \end{gather}
        $\int_y \left(\frac{D Y}{p(X)} - \frac{(1-D) Y}{1-p(X)} \right) \varphi(y,t^* ) dy$ and $\int_\zeta k''\left(\zeta \right)^2 d\zeta$ are bounded by Assumptions \ref{ass:sqintegr} and \ref{ass:extraonk}. Since by Assumption \ref{ass:rateofs} $n \sigma_n^3 \rightarrow \infty$, conclude that $\var(\hat{S}^2_n(t_n,\sigma_n)) \rightarrow 0$.
    \end{proof}

\subsection*{Theorem \ref{thm:asynormsmooth}}
\setcounter{thm}{3}
\begin{thm} 
    Consider the SWM policy $\hat{t}^s_n$ defined in equation \eqref{eq:estimator_sm} and the optimal policy $t^*$ defined in equation \eqref{eq:estimand}. Under Assumptions \ref{ass:identification}, \ref{ass:consistency} (with $s=h + 1$ for some $h\geq 2$), \ref{ass:nonflattau}, and \ref{ass:asymptotic_sm}, as $n \rightarrow \infty$:
    \begin{enumerate}
        \item if $n \sigma_n^{2h + 1} \rightarrow \infty$, $$\sigma_n^{-h}(\hat{t}^s_n - t^*) \rightarrow^p H^{-1}A;$$
        \item  if $n \sigma_n^{2h + 1} \rightarrow \lambda < \infty$, $$(n\sigma_n)^{\frac{1}{2}}(\hat{t}^s_n - t^*) \rightarrow^d \mathcal{N}(\lambda^{\frac{1}{2}}H^{-1}A, H^{-2}\alpha_2 K);$$
    \end{enumerate}
    where $A$, $\alpha_1$, and $\alpha_2$ are:
    \begin{align}
    A =& -\frac{1}{h!} \alpha_1 \int_y \left(Y_1 - Y_0 \right) \varphi^{h}_x(y,t^*) d y \\
    \alpha_1 =& \int_\zeta \zeta^h k'\left(\zeta \right) d \zeta \\
    \alpha_2 =& \int_\zeta k'\left(\zeta \right)^2 d \zeta.
    \end{align}
\end{thm}
\begin{proof}
    Consider a Taylor expansion of $\hat{S}^1_n(t,\sigma_n)$ about $t=t^*$:
    \begin{gather}
        \hat{S}^1_n(\hat{t}^s_n,\sigma_n) = \hat{S}^1_n(t^*,\sigma_n) + \hat{S}^2_n(\tilde{t},\sigma_n) (\hat{t}^s_n - t^*)
    \end{gather}
    with $|\tilde{t}-t^*| \leq |\hat{t}^s_n - t^*|$. By Theorem \ref{thm:cons_sm}, $\hat{t}^s_n \rightarrow^{a.s.} t^*$, and hence with probability approaching 1 $\hat{t}^s_n$ is an interior point of $\mathcal{T}$. It means that, with probability approaching 1, $\hat{S}^1_n(\hat{t}^s_n,\sigma_n) = 0$.

    To prove the first result of the theorem, note that with probability approaching one as $n \rightarrow \infty$
    \begin{gather}
        \sigma_n^{-h} \hat{S}^1_n(t^*,\sigma_n) + \sigma_n^{-h} \hat{S}^2_n(\tilde{t},\sigma_n) (\hat{t}^s_n - t^*) = 0.
    \end{gather}
    By lemmas \ref{lemma4} and \ref{lemma5}, $\plim_{n \rightarrow \infty} \hat{S}^2_n(t_n,\sigma_n) =- H$, and $H \neq 0$ by Assumptions \ref{ass:uniq} and \ref{ass:smoot}. Hence
    \begin{gather}
        \sigma_n^{-h}(\hat{t}^s_n - t^*) = H^{-1} \sigma_n^{-h} \hat{S}^1_n(t^*,\sigma_n) + o_p(1)
    \end{gather}
    and since $\sigma_n^{-h} \hat{S}^1_n(t^*,\sigma_n) \rightarrow^p A$ by lemma \ref{lemma2}, $$\sigma_n^{-h}(\hat{t}^s_n - t^*) \rightarrow^p H^{-1}A.$$

    Analogously, to prove the second result note that
    \begin{gather}
        (n\sigma_n)^{\frac{1}{2}} \hat{S}^1_n(t^*,\sigma_n) + (n\sigma_n)^{\frac{1}{2}} \hat{S}^2_n(\tilde{t},\sigma_n) (\hat{t}^s_n - t^*) = 0.
    \end{gather}
    with probability approaching 1, and hence
    \begin{gather}
        (n\sigma_n)^{\frac{1}{2}} (\hat{t}^s_n - t^*) = H^{-1} (n\sigma_n)^{\frac{1}{2}} \hat{S}^1_n(t^*,\sigma_n).
    \end{gather}
    Since $(n\sigma_n)^{\frac{1}{2}} \hat{S}^1_n(t^*,\sigma_n) \rightarrow^d \mathcal{N}(\lambda^{\frac{1}{2}}A,\alpha_2 K)$ by lemma \ref{lemma2}, $$(n\sigma_n)^{\frac{1}{2}}(\hat{t}^s_n - t^*) \rightarrow^d \mathcal{N}(\lambda^{\frac{1}{2}}H^{-1}A, H^{-2}\alpha_2 K).$$

    For the third result, first compute the asymptotic bias and the asymptotic variance of $\hat{t}^s_n - t^*$:
    \begin{align}
        \E[\hat{t}^s_n - t^*] =& -\lambda^{\frac{1}{2}}\frac{A}{H}(n\sigma_n)^{-\frac{1}{2}} =-\lambda^{\frac{1}{2}}\frac{A}{H} \lambda^{-\frac{1}{2(2h+1)}} n^{-\frac{h}{2h+1}} \\
        \var(\hat{t}^s_n - t^*) =& \frac{\alpha_2 K}{H^2} \lambda^{-\frac{1}{2h+1}} n^{-\frac{2h}{2h+1}} 
    \end{align}
    and then the MSE:
    \begin{gather}
        MSE = \frac{1}{H^2} n^{-\frac{2h}{2h+1}} \left[ \alpha_2 K \lambda^{-\frac{1}{2h+1}} + A^2 \lambda^{\frac{2h}{2h+1}} \right]
    \end{gather}
    which is minimize setting $$\lambda = \lambda^* = \frac{\alpha_2 K}{2hA^2 }.$$
\end{proof}

\subsection*{Corollary \ref{cor:smoot_regret}}
\begin{corollary} \label{cor:smoot_regret}
    Asymptotic distribution of regret $\mathcal{R}(\hat{t}^s_n)$ is:
    \begin{gather*}
        n \sigma_n \mathcal{R}(\hat{t}^s_n) \rightarrow^d  \frac{1}{2} \frac{\alpha_2 K}{H} \chi^2\left(1,\frac{\lambda A^2}{\alpha_2 K}\right)
    \end{gather*}
    where $\chi^2\left(1,\frac{\lambda A^2}{\alpha_2 K}\right)$ is a non-centered chi-squared distribution with 1 degree of freedom and non-central parameter $\frac{\lambda A^2}{\alpha_2 K}$.
    The expected value of the asymptotic distribution is:
    \begin{align} \label{eq:expregretswm}
        \frac{1}{2} \frac{\alpha_2 K}{H} \left(1 + \frac{\lambda A^2}{\alpha_2 K} \right) =  \frac{\alpha_2}{2} \frac{ K}{H} + \frac{1}{2} \frac{\lambda A^2}{H}.
    \end{align}
    Let $\sigma_n = (\lambda/n)^{1/(2h +1)}$ with $\lambda \in (0, \infty)$. The expectation of the asymptotic regret is minimized by setting $\lambda = \lambda^* = \frac{\alpha_2 K}{2hA^2 }$: in this case the expectation of the asymptotic distribution scaled by $n^\frac{2h}{2h+1}$ is $ A^{\frac{2}{2h+1}} K^{\frac{2h}{2h+1}} H^{-1} C^s$, where $C^s = \frac{2h+1}{2} \left( \frac{\alpha_2}{2h} \right)^\frac{2h}{2h+1}$ is a constant not dependent on $P$.
\end{corollary}

\begin{proof}
    Result in equation \eqref{eq:asympt_regret_d} for $\hat{t}^s_n$ implies
    \begin{gather*}
        n\sigma_n \mathcal{R}(\hat{t}^s_n) \rightarrow^d \frac{1}{2} W''(\tilde{t}) \left( (n\sigma_n)^{\frac{1}{2}} \left(\hat{t}^s_n-t^*\right) \right)^2.
    \end{gather*}
    where $|\tilde{t}-t^*| \leq |\hat{t}^s_n - t^*|$. By continuous mapping theorem
    \begin{gather*}
        W''(\tilde{t}) \rightarrow^p W''(t^*) = H
    \end{gather*}
    and hence by Slutsky's theorem
    \begin{gather*}
    n\sigma_n \mathcal{R}(\hat{t}^s_n) \rightarrow^d \frac{1}{2} H \left(\mathcal{N}(\lambda^{\frac{1}{2}}H^{-1}A, H^{-2}\alpha_2 K) \right)^2 =^d \frac{1}{2} \frac{\alpha_2 K}{H} \left(\mathcal{N}(\lambda^{\frac{1}{2}}A/\sqrt{D}, 1) \right)^2.
     \end{gather*}
     By definition, $\chi^2\left(1,\frac{\lambda A^2}{\alpha_2 K}\right) =^d \left(\mathcal{N}(\lambda^{\frac{1}{2}}A/\sqrt{D}, 1) \right)^2$, and $\E\left[\chi^2\left(1,\frac{\lambda A^2}{\alpha_2 K}\right)\right] = \left(1 + \frac{\lambda A^2}{\alpha_2 K} \right)$.

     When $\sigma_n = (\lambda/n)^{1/(2h +1)}$, the expectation of asymptotic regret is minimized by
     \begin{gather*}
         \lambda^* = \argmin_{\lambda} (n \sigma_n)^{-1} \left(\frac{\alpha_2}{2} \frac{ K}{H} + \frac{1}{2} \frac{\lambda A^2}{H} \right) = \argmin_{\lambda} \alpha_2 K \lambda^{-\frac{1}{2h+1}} + A^2 \lambda^{\frac{2h}{2h+1}}
     \end{gather*}
     which is solved by $\lambda^* = \frac{\alpha_2 K}{2hA^2 }$.

     By substituting $\sigma_n$ by $(\lambda/n)^{1/(2h +1)}$, and $\lambda$ by $ \frac{\alpha_2 K}{2hA^2 }$, the expectation of the asymptotic regret multiplied by $(n \sigma_n)^{-1}$ is
     \begin{gather*}
         (n \sigma_n)^{-1} \left(\frac{\alpha_2}{2} \frac{ K}{H} + \frac{1}{2} \frac{\lambda A^2}{H} \right) = n^{-\frac{2h}{2h+1}} A^{\frac{2}{2h+1}} K^{\frac{2h}{2h+1}} H^{-1} \frac{2h+1}{2} \left( \frac{\alpha_2}{2h} \right)^\frac{2h}{2h+1}
     \end{gather*}
     and the expectation of the asymptotic distribution scaled by $n^\frac{2h}{2h+1}$ is $A^{\frac{2}{2h+1}} K^{\frac{2h}{2h+1}} H^{-1} C^s$, where $C^s = \frac{2h+1}{2} \left( \frac{\alpha_2}{2h} \right)^\frac{2h}{2h+1}$ is a constant not dependent on $P$.
\end{proof}

\subsection*{Theorem \ref{thm:boots}}
\begin{thm}
Consider estimators $\hat{t}^e_n$ defined in equation \eqref{eq:estimator} and $\hat{t}^b_n$ defined in equation \eqref{eq:estboot} and the estimand $t^*$ defined in equation \eqref{eq:estimand}. Under Assumptions \ref{ass:identification}, \ref{ass:consistency} (with $s=2$), \ref{ass:asymptotic}, and \ref{ass:4mom}, as $\hat{H}_n \rightarrow^p H$ and $n\rightarrow \infty$,
\begin{gather}
    n^{1 / 3}\left(\hat{t}^b_n - \hat{t}_n\right) \rightarrow^d (2\sqrt{K}/H)^{\frac{2}{3}}\argmax_r \left(B(r) - r^2 \right)
\end{gather}
where the limiting distribution is the same as in Theorem \ref{thm:asyd}.
\end{thm}
\begin{proof}
The result follows from the main theorem in \cite{cattaneo2020bootstrap}. I will show that the assumptions for their results hold. My case is the benchmark case with $m_n = m_0$ (in my notation, $m$), $M_n = M_0$ (in my notation, $P m$) and $q_n =1$. Hence, in my case, class $\mathcal{M}_n$ coincides with $m$. I will verify the five conditions CRA:
\begin{enumerate}
    \item Consider envelope $F = 2\left|\frac{D Y}{p(X)} - \frac{(1-D) Y}{(1-p(X))}\right|$: Assumption \ref{ass:sqintegr} guarantees it is square integrable. \\
    Since $M_n = M_0$, $\sup_t |M_n(t) - M_0(t)| = 0$. Under Assumption \ref{ass:uniq}, $t^*$ is the unique maximizer of $m$, and hence $\sup_{t \neq t^*} m(t) < m(t^*)$.
    \item $t^*$ is an interior point of $\mathcal{T}$ by Assumption \ref{ass:uniq}. Assumption \ref{ass:smoot} with $s=2$ guarantees that $Pm$ is twice continuously differentiable in a neighborhood of $t^*$.

    \item I proved that this condition is satisfied in the proof of Theorem \ref{thm:asyd}, under Assumption \ref{ass:sqintegr}. In my notation, $\delta = R$, $\mathcal{D}_n^{\delta '} =\mathcal{G}_R$, $\bar d_n^{\delta '} = G_R$.

    \item Note that:
    \begin{align}
        \E[ G_R^4] & =\E \left[ \left(\frac{DY}{p(X)} - \frac{(1-D) Y}{(1-p(X))}\right)^4 \ind\{\left|X-t^*\right| < R\} \right]\\
        & = 2 R\E \left[ \left(\frac{DY}{p(X)} - \frac{(1-D) Y}{(1-p(X))}\right)^4 \Bigg| X =t^* \right] + o(1).
    \end{align}
    and that
    \begin{gather}
        \E \left[ \left(\frac{DY}{p(X)} - \frac{(1-D) Y}{(1-p(X))}\right)^4 \right] = 
        \frac{1}{p(X)^3} \E \left[ Y_1^4 \right] + \frac{1}{(1-p(X))^3} \E \left[ Y_0^4 \right].
    \end{gather}
    Let $R=O(n^{-\frac{1}{3}})$. It follows from Assumption \ref{ass:4mom} that
    \begin{align}
        n^{-\frac{1}{3}}\E[ G_R^4] =& 2  n^{-\frac{1}{3}} R \left( \frac{1}{p(X)^3} \E \left[ Y_1^4 \right] + \frac{1}{(1-p(X))^3} \E \left[ Y_0^4 \right] \right) + o(1) \\
        =& o(1).
    \end{align}
    The second part of assumption 4 is the same as the first part of assumption 5 in Theorem 1.1 in \cite{kim1990cube}. The only difference is that it must be valid for $t=t^*$ and $t$ in a neighborhood of $t^*$. Since Assumptions \ref{ass:sqintegr} and \ref{ass:smoot} with $s=2$ are valid also in a neighborhood of $t^*$, the argument provided before holds also here.
    \item The first part of this assumption is the same as assumption 6 in \cite{kim1990cube}, and the second part is the same as assumption 7 in  \cite{kim1990cube}. Under Assumptions \ref{ass:smoot} with $s=2$, and \ref{ass:asymptotic}, arguments provided before hold.
\end{enumerate}
CRA assumptions 1-5 by \cite{cattaneo2020bootstrap} are satisfied; hence, their results in Theorem 1 hold. It implies
\begin{gather*}
    n^{1 / 3}\left(\hat{t}^b_n - \hat{t}_n\right) \rightarrow^d Q(r)
\end{gather*}
where $Q(r) = Q_1(r) + Q_0(r)$, and $Q_1$ is a non degenerate zero-mean Gaussian process, while $Q_0(r) = -\frac{1}{2}r^2H$. Process $Q(r)$ is the same as in Theorem \ref{thm:asyd}.
\end{proof}

\end{document}